\documentclass[journal]{IEEEtran}
\ifCLASSINFOpdf
  \usepackage[pdftex]{graphicx}
  \graphicspath{{./Figs/}}
\else
\fi
\usepackage{amsmath}
 \usepackage{amsthm}
\newtheorem{remark}{Remark}
\newtheorem{definition}{Definition}
\newtheorem{proposition}{Proposition}

\newtheorem{theorem}{Theorem}
\usepackage{subfigure}

\usepackage[square,sort,comma,numbers]{natbib}

\begin{document}
%
\title{Energy control of  HVAC units for provable ancillary service provision
}
%
%
%

\author{Rupamathi~Jaddivada,~\IEEEmembership{Member,~IEEE,}
        Marija~D. Ili\'{c},~\IEEEmembership{Life~Fellow,~IEEE,}
        Mario~Garcia-Sanz,~\IEEEmembership{Member,~IEEE,}
        Mirjana~Marden,~\IEEEmembership{Member,~IEEE,}
        and ~ Sonja~Glavaski,~\IEEEmembership{Fellow,~IEEE}
\thanks{R. Jaddivada and M.D. Ilic are with,
Massachusetts Institute of Technology, Cambridge,
MA, 02139 USA e-mail: rjaddiva@mit.edu, ilic@mit.edu}
\thanks{ M. Garcia-Sanz and M. Marden are with Advanced Research Project Agency (ARPA-E), US Department of Energy, Washington, DC 20585 USA email:   marden\_mirjana@bah.com, mario.garcia-sanz@hq.doe.gov}
\thanks{S. Glavaski is currently with Pacific Northwest National Laboratory, 902 Battelle Boulevard, Richland, WA USA email:sonja.glavaski@pnnl.gov }
}

\maketitle

\begin{abstract}
In this paper, we consider the problem of controlling power consumption dynamics of residential heating, ventilation and air conditioning (HVAC) units so that they  follow the grid-side power specifications.
In order to do so, we design a novel dynamical energy controller which ensures regulation of the cumulative effects of  power imbalances. 
For this, 
we  derive a novel energy-based model that relates the HVAC physics-based dynamics to both real and reactive power balance at the point of interconnection with the grid. 
In contrast to several other approaches in the literature, we show that  a  limited number of HVAC units can  meet the stringent performance metrics set by the ARPA-E/NODES program on following the frequency  regulation signal, while maintaining consumer comfort. 
Theoretical and simulation-based model and control validation is provided by making use of real-world HVAC data. 
\end{abstract}

\begin{IEEEkeywords}
Demand response, Load flexibility, Themostatically controlled loads, Residential HVAC loads, Energy-based modeling and control
\end{IEEEkeywords}

%
\IEEEpeerreviewmaketitle

\section{Introduction}
Renewable Portfolio Standards (RPS) being adopted by several states in the US lead to large-scale deployment of solar panels and wind farms \cite{RPS}. These technologies are often not controlled by the grid operator and thus result in hard-to-predict disturbances in the grid. 
Shown in Fig. \ref{fig:CalISODuck} is the projected hard-to-predict in near real-time net load change in California over the next several years \cite{CalISODuck}, leading to increased reliance on ancillary services. 
\textit{Ancillary service}  includes frequency control, spinning reserves, and operating reserves, all related to additional power adjustments needed to offset the mismatch between the net demand and scheduled generation. 
Traditionally, large conventional generation units have been utilized for providing ancillary services. 
A major operational challenge is that these units have significant inertia and thus incur increased wear and tear, as they follow the net load patterns, such as the one indicated in Fig. \ref{fig:CalISODuck}.  

Therefore, utilities are exploring the potential of the flexibility provided by the residential consumer-end resources distributed throughout the grid, such as the electric vehicles, batteries and thermostatically controlled loads (TCLs) like electric water heaters, air conditioners and refrigeration units \cite{callaway2009tapping,koch2011modeling,qin2014modeling,alizadeh2014scalable}.  
This paper's particular interest is the flexibility offered by the heating, ventilation, and air conditioning (HVAC) units. They comprise 40\% of the total residential power demand and thus present an enormous potential for providing various types of ancillary services to the grid \cite{doe2015quadrennial}. 
Furthermore, their inherent thermal storage can be leveraged to precisely modify the power consumption and provide ancillary services while still meeting the desired temperature requirements of the end-user. 

\begin{figure}[!htbp]
	\centering
	\includegraphics[width= 1.0 \linewidth]{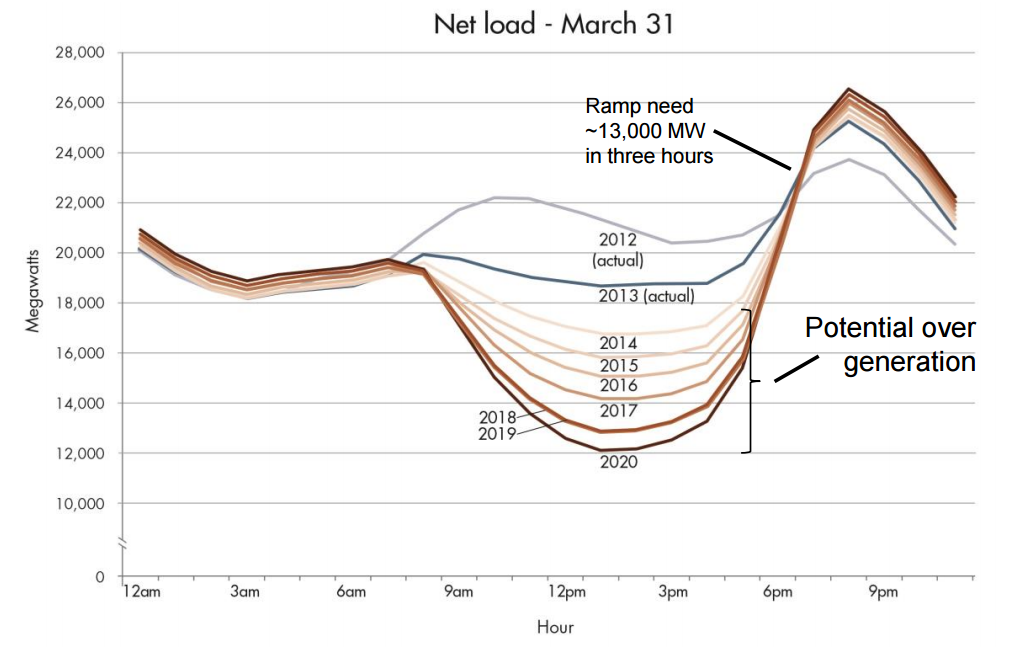}
	\caption{The duck curve showing steep ramping needs and over generation of wind farms in California \cite{CalISODuck}}
	\label{fig:CalISODuck}
\end{figure}

Different control methods broadly categorized as direct and indirect load control can enable  HVACs  to participate in ancillary services. In the direct load control method, the utility or the aggregator models many HVAC units and interrupts the power signals  of several  HVAC units to obtain desired aggregate performance. 
This method requires accurate prediction of the thermal flexibility of each HVAC unit to have a reliable supply of ancillary services at an aggregate level 
\cite{callaway2009tapping,zhang2013aggregated,almassalkhi2017packetized}. 
Through the indirect load control method, the consumer or the appliance automatically adjusts its switching cycles in response to real-time electricity prices  or frequency deviations in power systems, thereby providing ancillary services indirectly \cite{lu2004state,subbarao2013transactive,shu2012dynamic}. 
The relationship between the HVAC temperature setpoints and the real-time pricing is often complicated. Thus indirect load control method results in deviations from the power schedules, especially when the cluster of loads is not diverse enough or if a significant fraction of them is operating at their limits \cite{hao2014aggregate,lu2012evaluation}. 
Some other control methods combining these different methods include passivity-based distributed optimization \cite{chinde2017passivity,mukherjee2012building,hatanaka2017integrated}, multi-stage robust optimization considering detailed HVAC models \cite{qureshi2018hierarchical}, two layered-control distributing the decision-making between local and global controllers \cite{navidi2018two}. 
However, these approaches do not consider the practical limitations of the controllers or the availability of sensor measurements. 
For example, the HVAC needs to remain in the OFF state for a specific length of time before switching it back ON. These are called short-cycling constraints, which have been considered to some extent in \cite{zhang2013aggregated,hao2014aggregate,sanandaji2015ramping}. 
More importantly, to allow for the participation of a cluster of loads, the NODES program of ARPA-E has indicated the performance metrics in terms of the response time, ramp time, and availability of units explained in Section 2 \cite{ARPAEFOA}. However, most methods proposed in the literature can meet these metrics only upon considering a large number of units that are diverse enough \cite{hao2014aggregate,lu2012evaluation,bernstein2021final}. 
Thus in this paper, we consider effects of a limited number of  HVAC units and assess their capability in following the scaled-down regulation signal, all while ensuring provability. By provability we mean that the approach needs to result in reproducible results with theoretical performance guarantees. 

In Section \ref{Sec:Overview}, we 
review conventional HVAC models, and identify sub-objectives needed to pose the control problem.
In Section \ref{Sec:Modeling}, we propose a novel energy-space model for the HVAC system 
building upon our recent work
\cite{ilic2018fundamental,ilic2018multi,ilic2019exergy,ilic2020unified}. 
Utilizing this model, we next summarize our proposed 
multi-layered energy-based control of HVAC in Section \ref{Section:Control}.
This control comprises a device-level sliding mode control tracking an output variable in energy-space to a reference value, explained in Section \ref{Sec:PrimaryControl}. 
Novel closed-loop droop relations 
mapping the controlled output, reference signals and the regulation signals  have been derived next in Section \ref{Sec:SecondaryControl}. These are further utilized in 
a slower time-scale MPC-based controller to maximize the efficiency and meet the regulation signals.
We finally conclude  by summarizing contributions of this paper and  suggest  directions for future research in Section \ref{Sec:Conclusions}.

\section{Problem formulation}
\label{Sec:Overview}
In this section, we review residential HVAC modeling along with numerous constraints imposed by the consumer, HVAC manufacturer and the utility operator for it to qualify as a reliable source of regulation. 

\subsection{Conventional models of  residential HVAC system}
HVAC system comprises the space or zone to be heated/cooled and the auxiliary electrical equipment that injects hot/cool air into the space. As a result, we can characterize two types of state variables, as shown in Fig. \ref{fig_HVACOL}.
\begin{figure}[ht]
	\centering 
	{\includegraphics[width=0.9\linewidth]{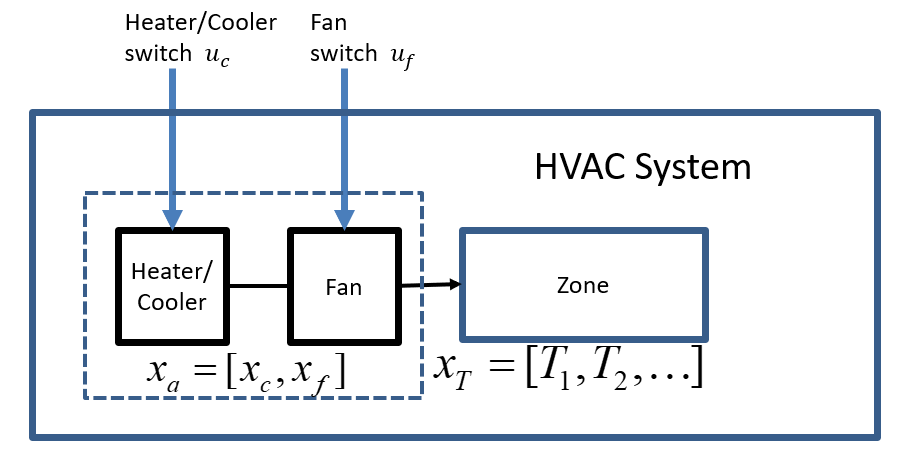}}
	\caption{Schematic of a HVAC system in open loop: The space is modeled using temperatures of multiple zones denoted using vector $x_T$. There are additional auxillary states $x_a$ corresponding to the fan and heater/cooling unit. The blower fan switch and the cooler/heater switch $u_c$ are the control inputs}
	\label{fig_HVACOL}
\end{figure}
The first set of variables characterize the temperatures of multiple zones,  representing the different rooms, walls and are denoted as $x_T$ \cite{deng2010building}. 
The electrical equipment comprises a fan that circulates air into the zones, while the heater/cooling unit maintains the temperature of the supply air. Let us denote the respective state variables using $x_a$. 
The models of these sub-systems are not exactly known. 

\textbf{Objective 1:}
The control design should be robust to model and parameter uncertainties and be 
simple enough to utilize limited available sensor measurements.

In residential HVAC systems, often single zone temperature is utilized to characterize the thermal load \cite{sanandaji2015ramping}. 
Furthermore, the power rating of the fan is orders of magnitude smaller than that of the compression system or arc furnace for cooling or heating, respectively. Therefore, in the rest of the paper, we 
only consider single control input, which is the switch position of the cooling unit, denoted as $u$. 
\begin{subequations}
\begin{align}
    C \dot{T} &= -\frac{1}{R} (T-T_0) + {P^{rated} u}
\label{Eqn_ConvModel} \\
 P^{rated} &= \dot{m}_a C_p (T_{sup} - T)   
 \label{Eqn_Psupply}
\end{align}
\end{subequations}
$T$, $T_0$ and $T_{sup}$ respectively denote the zonal temperature, ambient temperature and the supply temperature that is maintained by the compression system in heater/cooler block. 
$C, R$ represent the thermal capacitance and resistance, respectively. 
$P^{rated} u$ is the heating or cooling rate for with $u$ being the switch position taking a positive or negative value respectively. 
$\dot{m}_a$ is the airflow rate, $C_p$ is the specific heat capacity of the air. 
$T_{sup}, \dot{m}_a$ form the auxiliary state variables in Fig. \ref{fig_HVACOL} and are typically assumed to be constant. 

\textbf{Objective 2:} 
It is required to ensure that the temperature reaches the setpoint $T^{ref}$ as set by the consumer with a tolerable temperature deviations of $T_{db}$, i.e. $T(t) \in \left[T^{ref} - T_{db}, T^{ref} + T_{db}\right] \forall t$ 

The pre-programmed embedded automation $u_{em}$  in today's HVACs is based on PID controllers and satisfies Objective 2 under all circumstances. 
In context of Eqn. \eqref{Eqn_ConvModel}, 
$u_{em}$ is interpreted as a bang-bang control, with the following logic assuming the heating mode of operation. 
\begin{subequations}
\begin{align}
    u = 
    u_{em}(t) &= 1 \quad \text{if} \qquad T(t) < T^{ref} - T_{db}\\
    &=0 \quad \text{if} \qquad T(t) > T^{ref} + T_{db}\\
    &=u_{em}(t-\delta t) \qquad \text{otherwise}
\end{align}
\label{Eqn_SwitchConv}
\end{subequations}
Here, $\delta t$ is the continuous time implementation timestep which is dictated by the bandwidth of the available sensor measurements. 

\textbf{Objective 3:}
 Several considerations from the perspective of the control designer and manufacturer
 respectively must be met as follows:  
\begin{itemize}
    \item The control needs to be implemented using the knowledge of temperature and electrical power consumption alone. This is because sensors for other internal variables are unavailable and are expensive. 
    \item The compressor needs to be left ON at least for 5 minutes and it can remain in that state at most for 15 minutes. 
\end{itemize}

\subsection{Ancillary services}
The aggregator or service provider 
acts as the middle man between the individual HVAC units and the grid operator. 
The aggregator's responsibility is to provide the regulation service as commanded by the grid operator over 
minutes timescale as given by the discrete time samples $k$ shown in Figure \ref{fig:TimeScales}. 
However, there are thermal and electric power dynamic processes evolving at finer time granularity of seconds and milliseconds,  respectively. 
The controller actions need to be taken at these finer timescales to have their accumulated effect over minutes timescale track the commanded regulation signal. 

\begin{figure}[!htbp]
	\centering
	\includegraphics[width=  \linewidth]{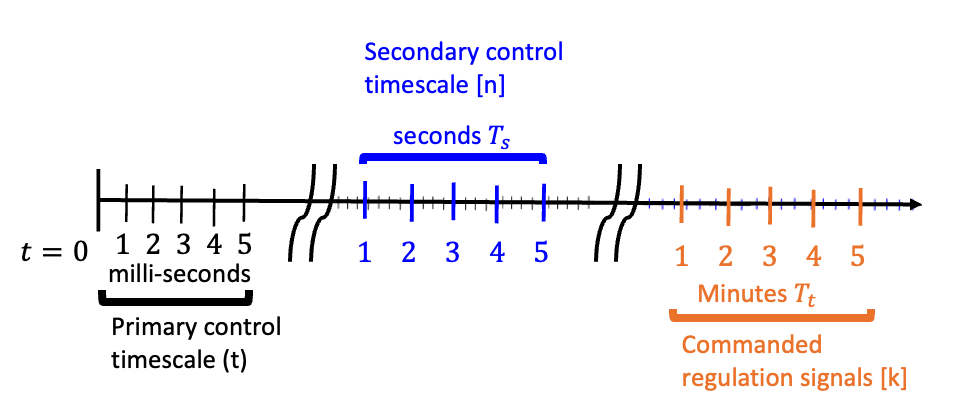}
	\caption{Multi-timescale evolution of control objectives: Aggregate regulation signals arrive every $kT_t$ instant; Individual HVAC unit commands need to arrive every $nT_s$ step while the power consumption dynamics evolve at a much faster continuous millisecond timescale $(t)$}
	\label{fig:TimeScales}
\end{figure}

\textbf{Objective 4:}
From the grid perspective, if single HVAC unit were to participate in ancillary services, 
it should consume  average power  every $kT_t$ instant equal to the 
commanded regulation signal $P^{reg}(kT_t)$ 
\begin{equation}
\sum_{t = (k-1)T_t}^{kT_t} {P(t)} = {P^{reg}}(k{T_t})
\label{Eqn_RegSingalCont}
\end{equation}

\subsection{ARPA-E performance metrics}
The utility performs a few tests to assess if the units willing to participate in ancillary services meet specific performance metrics before being deemed a reliable provider. 
There is no clear rationale for setting these performance metrics. 
These metrics are not standard even for current grid operations where controllable generation units primarily provide ancillary services. 
Following the same principles, ARPA-E has proposed performance metrics for different types of ancillary services as needed to be met by a group of small-capacity flexible units coordinated by and collectively referred to as 
Network Optimized Distributed Energy System (NODES)\cite{ARPAEFOA}. 
We interpret these metrics in the context of a single HVAC unit of interest in this paper. 

\textbf{Objective 5:}
The averaged power consumption over $T_t$ rate $P[k]$ of each HVAC unit should observe the following specifications as it is commanded to follow step change in regulation signal at time $t=0$, as shown in Fig. \ref{fig:PerformanceMetrics}.  
\begin{figure}[!htbp]
	\centering
	\includegraphics[width= 0.8 \linewidth]{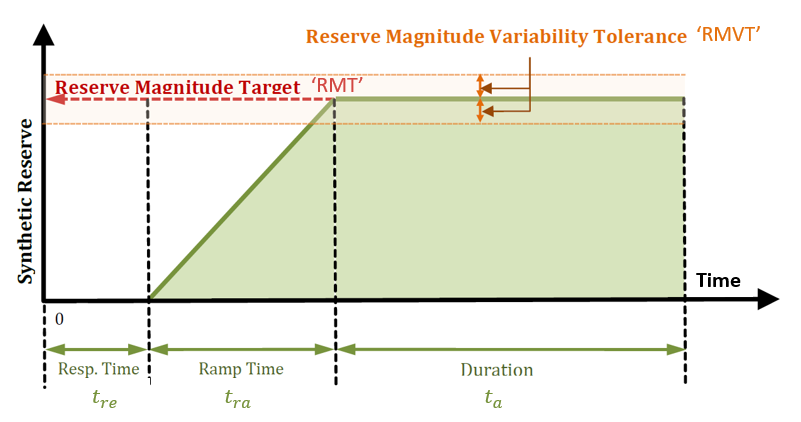}
	\caption{Performance metrics to be satisfied for participating in ancillary services \cite{ARPAEFOA}}
	\label{fig:PerformanceMetrics}
\end{figure}

\begin{itemize}
    \item Response time ($t_{\text{re}}$): 
    The time it takes for the actuators to respond to the commanded signals - less than 5 seconds. 
    \item Reserve magnitude target (RMT): 
    The maximum step change in regulation signal that the HVAC unit can track - 
    7\% of rated capacity \footnote{It is stated 7\% of the peak inflexible demand in the original specifications}. 
    \item Reserve magnitude variability tolerance (RMVT): 
    Maximum deviation between the actual power adjustments and the reserve magnitude target - 5\% of 'RMT'
    \item Ramp time ($t_{\text{ra}}$): 
    The maximum time it takes for the HVAC unit to reach the 'RMT' and stay within 'RMVT' deviations - 5 seconds 
    \item Availability time ($t_{\text{a}}$): 
    The period of time for which the reserves can be supplied - 3 hours
\end{itemize}

We ensure satisfaction of these performance metrics by explicitly considering some of these metrics in the control design as will be explained in Section \ref{Section:Control}.

\section{Energy-based modeling of HVAC}
\label{Sec:Modeling}
\subsection{Motivation for new modeling tools}
The thermal model in Eqn. \eqref{Eqn_ConvModel} is most commonly used in the literature. 
This equation is based on the first law of energy conservation principle. 
However, control designed based on the first law alone suffers major inconsistencies. 
For instance, the power absorbed by the HVAC unit is given by 
\begin{equation}
    P^{r,out} = \frac{1}{R}\left(T-T_0\right) + C\frac{dT}{dt}
\end{equation}
The first component is the useful work done by the HVAC unit and the second term is the rate of change of stored energy. 
On the other hand, the injected power into the HVAC unit 
for the embedded switching logic in Eqn. \eqref{Eqn_SwitchConv} 
is equal to
\begin{equation}
    P^{r,in} = P^{rated} u_{em}
\end{equation}
If the control were to satisfy the energy conservation principles, the two values of power interactions $P^{r,out}$ and $P^{r,in}$ should be equal. However, this is not the case as evidenced  by the 
the time trajectory plots in Figure \ref{fig_HVACPowerImbalance}. 
It may appear that this imbalance stems from the discrete control actions. However, the accumulated energy imbalance over time also doesn't average out to zero as shown in Figure \ref{fig_EnergyImbalance}. 
\begin{figure}[!htbp]
	\centering 
	{\includegraphics[width=0.7\linewidth]{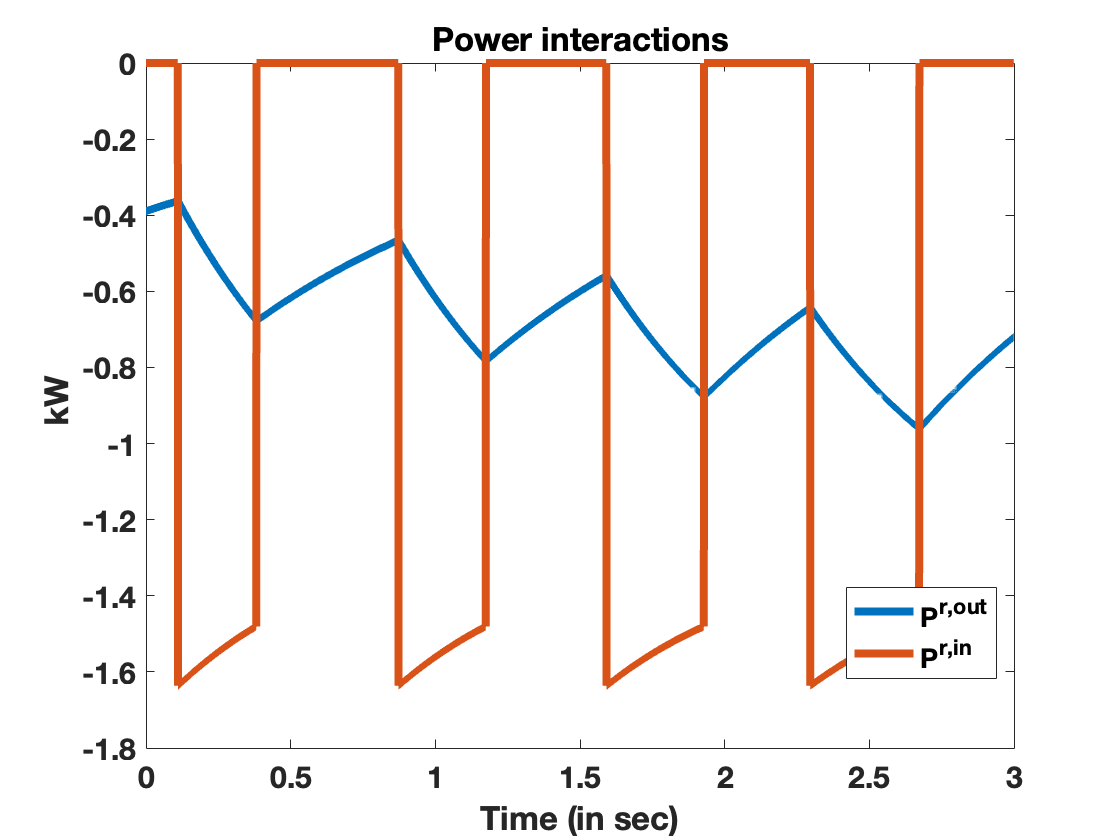}}
	\caption{The time trajectories of instantaneous power absorption and power injections seen by the HVAC at the electrical terminal}
	\label{fig_HVACPowerImbalance}
\end{figure}

\begin{figure}[!htbp]
	\centering 
	{\includegraphics[width=0.7\linewidth]{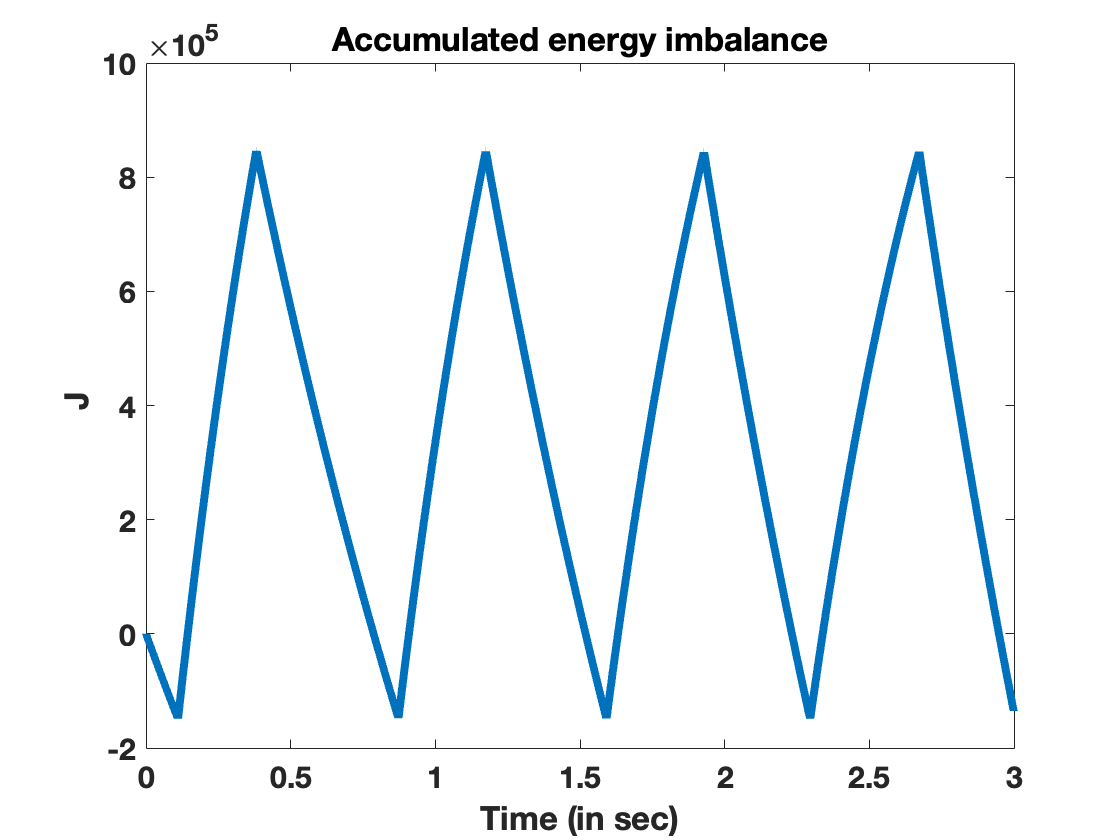}}
	\caption{Accumulated energy imbalance at the electrical terminal $\int_0^t \left(P^{r,out} - P^{r,in} \right)dt$}
	\label{fig_EnergyImbalance}
\end{figure}
This inconsistency stems fundamentally because the controllers based on first law of energy conservation principle alone do not align the rate at which power evolves thereby leaving a power imbalance offset that creates accumulated energy imbalance over time.  We thus introduce a novel energy modeling approach to explicitly accommodate the rate of change of power variables and control its  cummulative effects on temperature and comfort. 

\subsection{Energy modeling}
The schematic of energy interaction for the detailed HVAC model is shown in Fig. \ref{fig_HVACEnergyOLDetail}. 
The compressor system can operate either in heating/cooling mode that dictates the supply temperature $T_s$ into the zone, used in the simplified model in Eqn. \eqref{Eqn_ConvModel}. 
The Air Handling Unit (AHU) injects air at desired temperature into the zone, which circulates back through the fan. 
The zone is associated with thermal energy $U$ while the AHU unit is associated with stored energy $E_a$. The AHU comprises a blower fan, a and a compression unit or chiller. These can both be individually controlled. We assume that these units have internal automation. For enabling the HVAC units' participation in ancillary services, we incorporate an additional switch shown in Fig. \ref{fig_HVACEnergyOLDetail} that modulates the rate at which power enters the HVAC unit. It is important to note that we do not alter the existing power consumption as needed by the chiller and fan while controlling the rate of reactive power injections of the switching control action. 

\begin{figure}[!htbp]
	\centering 
	{\includegraphics[width=1.0\linewidth]{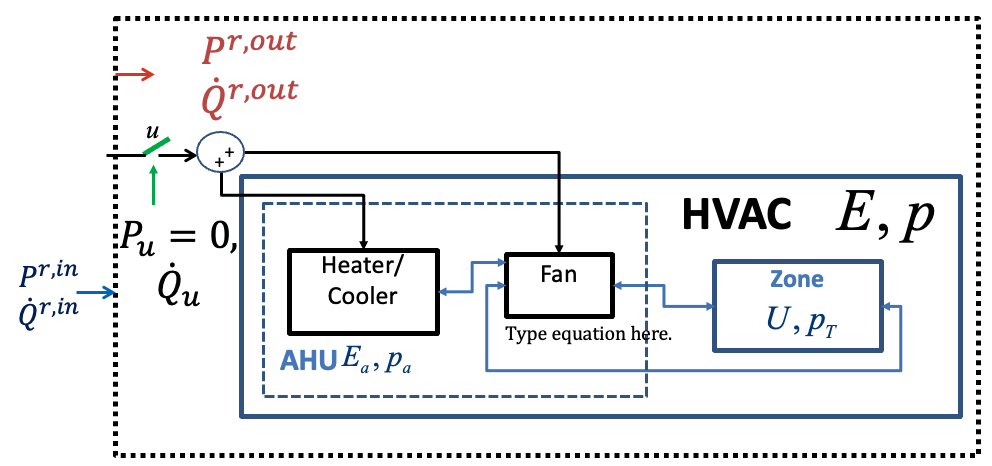}}
	\caption{Schematic of HVAC system in energy space. 
	The overall system is associated with stored energy $E$ and its rate of change $p$ which can be controlled by the switch shown in green. The switch itself injects non-zero $\dot{Q}_u$, resulting in net absorbed HVAC power variables $P^{r,out}$, $\dot{Q}^{r,out}$, while the grid power injections at instantaneous time are $P^{r,in}, \dot{Q}^{r,in}$}
	\label{fig_HVACEnergyOLDetail}
\end{figure}
All the interactions between subsystems can get overly complex if modeled in conventional state space. 
We instead assess the overall energy conversion dynamics, and strive to design the control without exactly knowing the internal physical models. 

\subsection{Preliminary definitions of variables in energy space}
In any energy domain, power variable are defined using the port variables called effort flow variables. 
These pairs are (voltage, current) in electric domain, (pressure, volume flow rate) in fluid domain and similarly based on Eqn. \eqref{Eqn:IntEnergy}, they are (Temperature, entropy flow) in thermal domain. 
The time integral of this quantity is defined as the energy injected into the component.
We start by defining thermal energy. 
\begin{definition}
	The internal or thermal energy of a thermodynamic fluid (air in the case of HVAC) can be expressed as follows: 
	\begin{equation}
	U = \int_0^t T dS = \int_0^t T \frac{dS}{dt} dt =   \int_0^t T S_f dt
	\label{Eqn:IntEnergy}
	\end{equation}
	Here, $dS$ is the incremental entropy of the medium, $S_f = \dot{S}$ is the net entropy flow into finite volume at a temperature $T$. 
	 \label{Defn:Energy}
\end{definition}

\begin{definition}
    The energy in tangent space is defined by replacing the effort-flow variables in definition of stored energy with that of  its derivatives. 
    \begin{equation}
	U_t = \int_0^t \frac{dS_f}{dt} \frac{dT}{dt} dt 
	\label{Eqn:EtEnergy}
	\end{equation}
\end{definition}

Irrespective of energy domain, each arrow in Fig. \ref{fig_HVACEnergyOLDetail} models the power variables defined next.	
\begin{definition}
	Given the generalized effort $e$ and flow variables $f$ at a port, the instantaneous power and rate of change of generalized reactive power absorbed by the port is defined as in Eqn. \eqref{Eqn_RealReacP} \cite{ilic2018multi,wyatt1990time}.
	\begin{subequations}
		\begin{align}
		P &= ef\\
		\dot{Q} &= e\frac{df}{dt} - \frac{de}{dt}f 
		\end{align}
	\end{subequations}
	\label{Eqn_RealReacP}
\end{definition}
We have defined here the thermal energy $U$ and thermal energy in tangent space $U_t$. The similar characterization for electro-mechanical energy domains is explained in \cite{ilic2018fundamental,ilic2018multi} and other energy domains in \cite{ilic2019exergy}. 
The total energy of the HVAC system $E$ is the sum of the stored energy of the sub-systems in AHU $E_a$ and thermal energy $U$.
We have shown in \cite{ilic2018multi} that the dynamics of energy exchanges of a component with the rest of the system can be captured by modeling the dynamics of aggregate variables energy $E$ and its rate of change $ p = \dot{E}$. 
The model was derived for electrical circuits and was proven to hold for complex electromechanical systems by extending the definitions of energy space variables through effort-flow analogy for multiple energy domains. 
In the same way, extending the analogy to the hydraulic energy domain with a continuum of space, we have derived models for turbo-machinery in this energy space in \cite{ilic2019exergy}. We now extend the modeling approach to also capture thermal interactions within the HVAC system. 

\subsection{Stand-alone HVAC model in energy space}
\label{Sec:HVACEnergyModel}
We now utilize the definitions of energy space variables to show that the previously proposed energy space model in \cite{ilic2018multi} also holds for thermal  interactions of HVAC.
\begin{proposition}
	The stand-alone HVAC component model shown in Figure \ref{fig_HVACEnergyOLDetail} with an added switch at the interface exhibits the following general dynamics in energy state space comprising the states $\left[E,p\right]$
	\begin{subequations}
	\begin{align}
	\dot{E} &= - \frac{E}{\tau} + P^{r,out} = p \label{Eqn_GeneralEnergyModel1}\\
	\dot{p} &=  4E_t -  \dot{Q}_u - \dot{Q}^{r,out}\label{Eqn_GeneralEnergyModel2}
	\end{align}
	\label{Eqn_GeneralEnergyModel}
\end{subequations}
Here, 
$\frac{E}{\tau}$ represents the total damping losses. $E_{t}$ is the total stored energy in tangent space. 
$\dot{Q}_u$ is the controllable reactive power injected by the switch. The resulting internal state dynamics result in instantaneous power and rate of reactive power absorption $P^{r,out}$ and $\dot{Q}^{r,out}$ respectively. 
\end{proposition}
\begin{proof}
Starting from the model in Eqn. \eqref{Eqn_ConvModel}, we show that the aggregate energy space model holds in Appendix \ref{Sec:AppProof}. 
\end{proof}

\begin{remark}
	Here ${p} = P^{r,out}$ for when the damping losses $\frac{E}{\tau} = 0$. In such cases $p$ is the supplied power. If there were no switching reactive power injected,  $\int_0^t4 E_t dt$ is the power corresponding to the potential to do useful work while the 
	${Q}^{r,out}$
	represents component of power corresponding to wasted work \cite{ilic2018multi}. This wasted works is resulting from the dynamic inefficiencies associated with the energy conversion processes at instantaneous time. 
	It must be noted that this inefficiency is different from the damping losses from linear frictions, viscosity etc, which is captured in the term $\frac{E}{\tau}$. This dynamic inefficiency arises due to the difference in the rate at which thermal power gets absorbed by the zone and the rate at which the power is absorbed from the grid. 
	It is only when $2 \dot{Q}^{r,out}$ approaches $\dot{Q}_T$, maximum physical efficiency is achieved. 
\end{remark}

\subsection{Validation of energy model using real data} 
In order to keep the analytical derivations simple, we have assumed in Proposition 1 negligible stored energy in AHU when compared to that of the zone. However, the energy space model postulated holds even if this assumption were relaxed. 
To this end, we utilize real-world measurements of a detailed HVAC system to validate our proposed energy model. 

We performed a model validation experiment in \cite{wu2021validity} where the interface variables of chiller and pumps in compression system and that of fans and zones were used to conduct parameter estimation of the respective conventional state space models. These parameters are the inertial and damping coefficients of each of the components within a HVAC system. 
These parameters were used to compute the overall stored energy $E$, damping $\frac{E}{\tau}$ and the stored energy in tangent space, given the measurements of internal states (temperature of zone and chiller in heater/cooler units, air flow rate in the fan). The detailed definition of these energy variables is outside the scope of this paper. 
The reader is referred to \cite{wu2021validity}. 
We further utilized the measurements of voltages and currents at electrical terminals of compressors and fans to obtain instantaneous and rate of reactive power using Eqn. \eqref{Eqn_RealReacP}.

Shown in Figure \ref{Fig:ValidationZoomedOut} is the perfect match in the comparison plots of left hand side and right hand side of the proposed energy model of Eqn. \eqref{Eqn_GeneralEnergyModel} over longer timescales. 
Shown in Figure \ref{Fig:ValidationZoomedIn} are the same plots zoomed-in over finer time granularity. 

\begin{figure}[!htbp]
	\centering
	\subfigure[Zoomed-out Eqn. \eqref{Eqn_GeneralEnergyModel1}]{\includegraphics[width=0.785\linewidth]{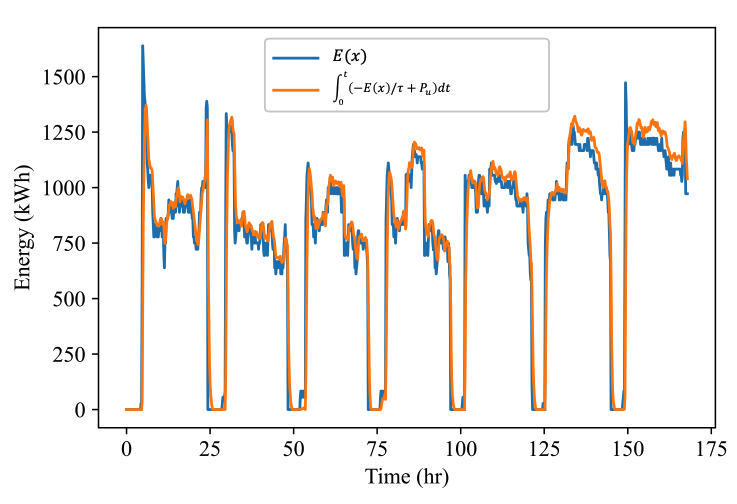}\label{Fig:ValidationZoomedOut1} }
	\subfigure[Zoomed-out Eqn. \eqref{Eqn_GeneralEnergyModel2}]{\includegraphics[width=0.785\linewidth]{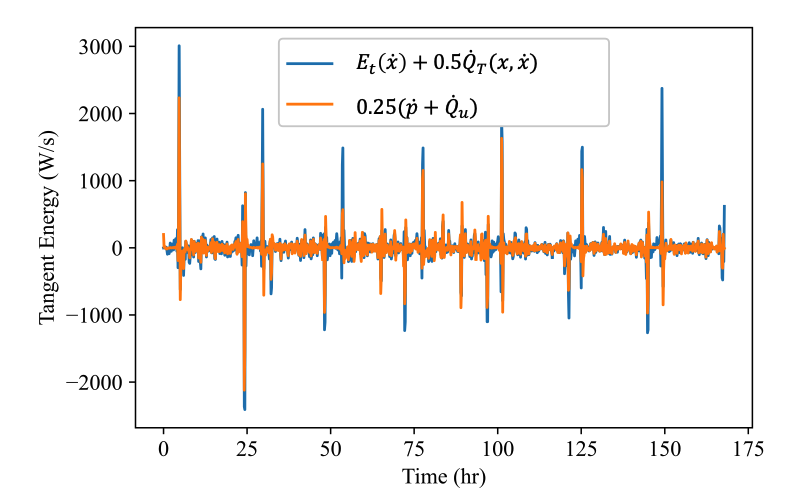}\label{Fig:ValidationZoomedOut2}}
	\caption{Zoomed-out validation of unified energy model using real-world data}
	\label{Fig:ValidationZoomedOut}
\end{figure}

\begin{figure}[!htbp]
	\centering
	\subfigure[Zoomed-in Eqn. \eqref{Eqn_GeneralEnergyModel1}]{\includegraphics[width=0.785\linewidth]{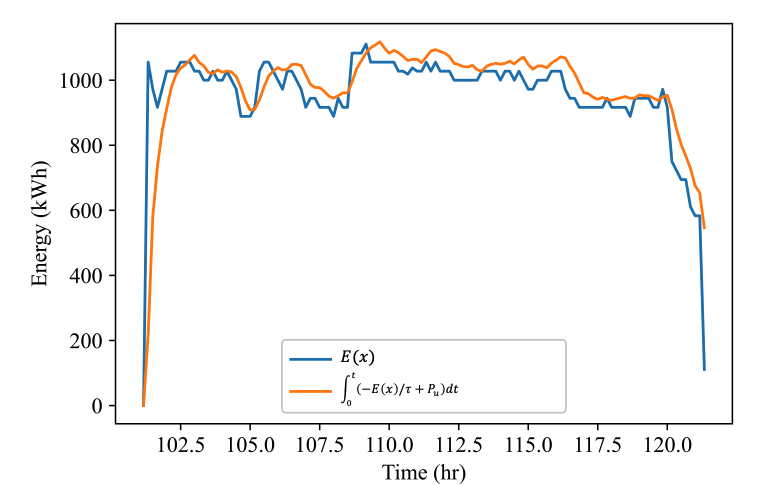}\label{Fig:ValidationZoomedIn1} }
	\subfigure[Zoomed-in Eqn. \eqref{Eqn_GeneralEnergyModel2}]{\includegraphics[width=0.785\linewidth]{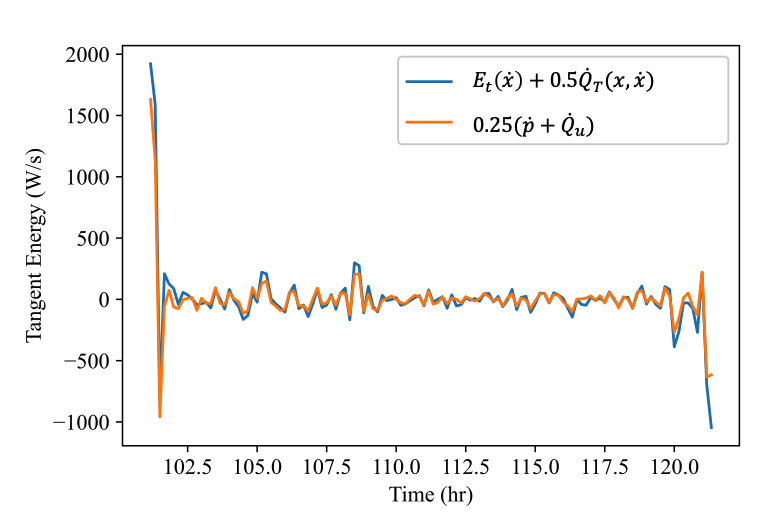}\label{Fig:ValidationZoomedIn2}}
	\caption{Zoomed-in validation of unified energy model using real-world data}
	\label{Fig:ValidationZoomedIn}
\end{figure}


\subsection{Interactive interconnected model}
There exists a unique transformation of state and state derivatives called interaction variable $z^{r,out}$ that has important structural properties. The rate of change of interaction variable is zero when disconnected from the grid. For details, see \cite{jaddivada2020unified}. 
These outgoing interaction variables are defined as: 
\begin{equation}
    	\dot{z}^{r,out} = \left[ \begin{array}{l}
{P^{r,out}}\\
{{\dot{Q}}^{r,out}}
\end{array} \right] = \left[ \begin{array}{l}
p + \frac{E}{\tau }\\
4{E_t} - \dot{Q}_u- \dot p
\end{array} \right] = \phi(\tilde{x},\dot{\tilde{x}})
\label{Eqn:OutIntVar}
\end{equation}
The definitions for $p$, $E_t$, $\dot{Q}_T$ for the simple dynamical model in Eqn. \eqref{Eqn_ConvModel} are introduced in Appendix \ref{Sec:AppProof}. As a result the outgoing interaction variable is abstracted as a function of extended states and its derivatives which is defined as $\tilde{x} = \left[x,u\right]^T$. Here $x$ is a vector of all HVAC state variables and $u$ is the switch control introduced for ancillary service participation in Figure \ref{fig_HVACEnergyOLDetail}. 

On the other hand, there is also incoming interactions into the HVAC from the electric grid due to the energy conversion dynamics of rest of the components connected to the grid i.e.
\begin{equation}
 \dot{z}^{r,in} = [P^{r,in}, \dot{Q}^{r,in}]^T 
\end{equation}
$\dot{z}^{r,out}$ is created through thermal processes, while $\dot{z}^{r,in}$ is the result of faster electric processes on the grid side. 
These two quantities need to align with each other for system interconnection to be stable and feasible \cite{ilic2018fundamental,jaddivada2020unified}. 
During transients this difference is not zero as shown for the instantaneous power component with conventional control design in Figure \ref{fig_HVACPowerImbalance}.
Driving this imbalance to zero is set as the primary control objective as will be explained in the next section. 
Notably, the energy space model in Eqn. \eqref{Eqn_GeneralEnergyModel} is linear in terms of the energy states $x_z = \left[E,p\right]^T$ where the control input in energy state space is $u_z = \dot{Q}_u$ and the term $\left(4E_t - \dot{Q}^{r,out}\right)$ is the state-dependent disturbance entering the model. Setting bounds on this disturbance allows us to perform linear control design in energy state space with provable design. 


\section{Proposed multi-layered energy-based control}
\label{Section:Control}
We exploit the linear nature of energy model to first design a provable primary control that addresses all the stated objectives. Shown in Figure \ref{fig:HVAC_Prop} is the proposed multi-layered control architecture based on the timescales illustrated in Figure \ref{fig:TimeScales}.
\begin{figure*}
	\centering
	\includegraphics[width= 0.9 \linewidth]{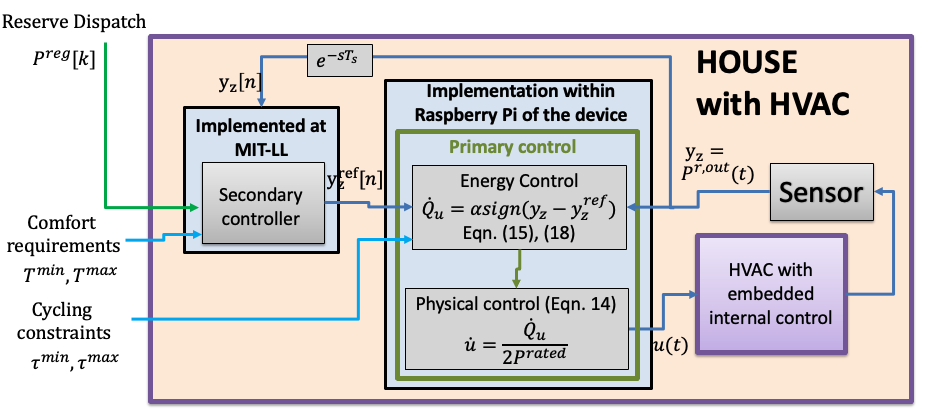}
	\caption{Proposed multi-layered information exchange for control of HVAC. Interface power measurements are utilized for sliding mode energy control at continuous timescales to track a reference value $y_z^{ref}[n]$ obtained over discrete $T_s$ timesteps such that the resulting HVAC power adjustments could track a given regulation signals over even longer $T_t$ timescale.}
	\label{fig:HVAC_Prop}
\end{figure*}

In particular, at fastest timescales, we design a sliding mode control in energy state space that relies on interface power measurements alone. 
The sliding mode design inherently allows for the design to be robust to model and parameter uncertainties. 
We further introduce an error threshold for the sliding mode design to explicitly incorporate the compressor cycling constraints. We can thus satisfy \textbf{objectives 1 and 3}. 

In context of the timescales illustrated in Figure \ref{fig:TimeScales}, we implement the the sliding mode design to have a carefully selected output variable in energy state space $y_z(t)$ track a consistence reference value $y_z^{ref}[n]$ varying at minutes timescale. 
Upon application of such design, it is possible to characterize a quasi-static droop relations between output of interest in energy state space, its reference value and the 
averaged power absorbed by the HVAC unit over minutes timescale. 
This reference value is the secondary control action designed in an interactive manner by accounting for the closed loop HVAC dynamics with the primary control and the thermal comfort constraints. As a result, we satisfy \textbf{objective 2}  over seconds timescale while also satisfying \textbf{objective 4} over minutes timescales. 

The sliding mode design over fastest timescales together with model predictive control over slower timescales allow us to ensure stability, regulations and incorporation of ARPA-E performance metrics stated in \textbf{Objective 5} in a provable manner. Meeting the performance metrics, especially the RMT depends on how large the target is, the thermal load of the HVAC unit and the number of HVAC units available. With our proposed approach, it is possible to atleast identify the extent to which the reserve magnitude target can be met apriori while satisfying all other performance metrics. This allows for bringing in additional HVAC units to thereby satisfy the reserve magnitude target as needed. 
Each of these steps will be explained in more detail in the subsequent sections. 


\section{Provable primary control}
\label{Sec:PrimaryControl}
\subsection{Control design in energy state space}
HVAC systems already have an embedded proprietary digital control. We propose to overwrite these signals through a robust primary control based on sliding mode energy control design.
Consider the energy space model in Eqn. \eqref{Eqn_GeneralEnergyModel} and let $\dot{Q}_u$ be the available degree of control and let us select the output variable of interest as in Eqn. \eqref{Eqn:Output}. 
\begin{subequations}
\begin{equation}
    y_z = P^{r,out}\left(E,p=0\right) = \frac{E}{\tau}= 
    \frac{1}{R}\left(T-T_0\right) \label{Eqn:Output}
\end{equation}
This quantity represents the device's instantaneous power absorption after the settling of stored energy dynamics. 
We set the reference value 
as follows: 
\begin{align}
    &y_z^{ref} = P^{r,in,ref} = P^{r,out,ref} + \Delta P^{r,out}[n]\notag\\ &=\frac{1}{R}\left(T^{ref} - T_0\right) + \Delta P^{r,out}[n]  \forall t \in [(n-1)T_s, (nT_s)]\\
    \label{Eqn:yzRef}
\end{align}
Here, the first term is the desired baseline power consumption to satisfy the thermal comfort requirements while the second term indicates the power consumption adjustments for participation in ancillary services. 
$y_z^{ref}[n]$ is the feed-forward signal that enters every $T_s$ timestep. The fast sliding mode design is expected to ensure $y_z \rightarrow y_z^{ref}$ before $t = nT_s$. 
To this end, we design a robust sliding mode control as in Eqn. \eqref{Eqn:ControlEnergy} that ensures finite settling time. 
\begin{equation}
    \dot{Q}_u = \alpha sign\left(\underbrace{y_z - y_z^{ref}}_{\sigma}\right)
    \label{Eqn:ControlEnergy}
\end{equation}
Here, $\sigma = y_{z} - y_z^{ref}$ represents the sliding surface, and $\alpha$ is the sliding mode gain. 
\label{Eqn:ControlContinuousTime}
\end{subequations}

\subsection{Physical control mapping}
Let us now zoom-in to the switch in Figure \ref{fig_HVACEnergyOLDetail} and characterize the reactive power flows as shown in Figure \ref{Fig:QdotSwitch}. Let us denote 
left hand side and right-hand side reactive power flows as $\dot{Q}_{sw,L},\dot{Q}_{sw,R}$. The switch on left hand side interacts with the electrical grid while that on right side interacts with the AHU and zone of the HVAC component. 
\begin{figure}[!htbp]
	\centering
	{\includegraphics[width=0.9\linewidth]{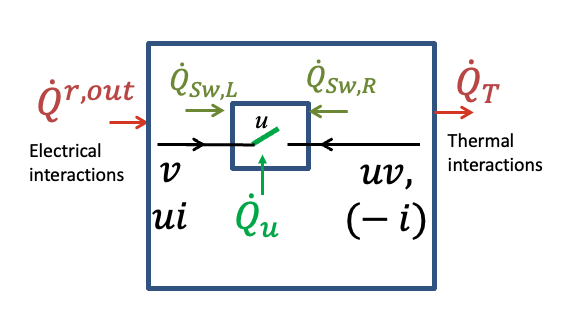} }
	\caption{Depiction of reactive power flows at the interfaces of the switch considered in Figure \ref{fig_HVACEnergyOLDetail}. 
	}
	\label{Fig:QdotSwitch}
\end{figure}
Given the electrical effort flow variables $\left(v,i\right)$ modulated through the switch $u$ at left and right hand side ports of the switch, we can utilize the definition of reactive power in Eqn. \ref{Eqn_RealReacP} to obtain the following expressions. 
\begin{subequations}
\begin{align}
    \dot{Q}_{u} &= \dot{Q}_{sw,L} - \dot{Q}_{sw,R}\\
    & = \left(v \frac{d}{dt}\left(ui\right) - \left(ui\right) \frac{dv}{dt}\right) 
     - \left(\left(uv\right) \frac{di}{dt} - i\frac{d}{dt}\left(uv\right)\right)\\
    & = 2 vi \dot{u} = 2P^{rated} \dot{u}
\end{align}
Here, $P^{rated}$ is the power flowing through the switch which is the same as the one defined in Eqn. \eqref{Eqn_ConvModel}. 
\label{Eqn:QudotDerivation}
\end{subequations}
Rearranging the terms in Eqn. \eqref{Eqn:QudotDerivation}, we obtain the dynamical switching control expression in Eqn. \eqref{Eqn:CtrlDynamics}. 
\begin{equation}
    \dot{u} = \frac{1}{2 P^{rated}} \dot{Q}_u 
    \label{Eqn:CtrlDynamics}
\end{equation}

\begin{theorem}
The closed loop model (Eqn. \eqref{Eqn_ConvModel}, \eqref{Eqn:CtrlDynamics}) exhibits following properties
\begin{enumerate}
    \item It is diffeomorphic to the energy state space model in Eqn. \eqref{Eqn_GeneralEnergyModel} 
    \item The virtual control design in Eqn. \eqref{Eqn:ControlEnergy} ensures finite settling time if the sliding mode gain is set equal to 
    $\alpha = \overline{L} + K$ where $\left|4E_t  - \dot{Q}^{r,out}\right| \le \overline{L}$, $ K>0$
    \item The objective of $y_z \rightarrow y_z^{ref}$ is achieved within a finite reaching time  $t_r \le \frac{K}{2|\sigma(0)|}$
    where $\left|\sigma(0)\right|$ represents the distance of operating point from the sliding surface at initial time. 
\end{enumerate}
\end{theorem}
\begin{proof}
The proof is provided in the Appendix \ref{Sec:ProofTheorem}. 
\end{proof}
The finite reaching time result of the sliding mode control is an important result. In particular, the sliding mode gain can be chosen to ensure the reaching time $t_r$ is at least ten times smaller than secondary control timescale $T_s$. We can further select $T_s$ to be equal to the ramp time $t_{ra}$ stated in Objective 5. 

For chosen values of $T^{ref} = 75 ^0$ F, $\Delta P^{r,out}[n] = 0$  and time varying ambient  temperature changes, 
we obtain finite settling time of the temperature to the reference value as shown in Figure \ref{fig:SimulinkSim1}.
The simulation here has utilized 
a more detailed residential HVAC simulation model with AHU dynamics simulated and with imperfect parameter corresponding to the damping losses being utilized. 
Notice that the the energy control design utilizes measurements of just the interface electrical power measurements and the zonal temperature values which are typically available in the existing residential smart meters and thermostats. 

\begin{figure}[!htbp]
	\centering
	\includegraphics[width=  \linewidth]{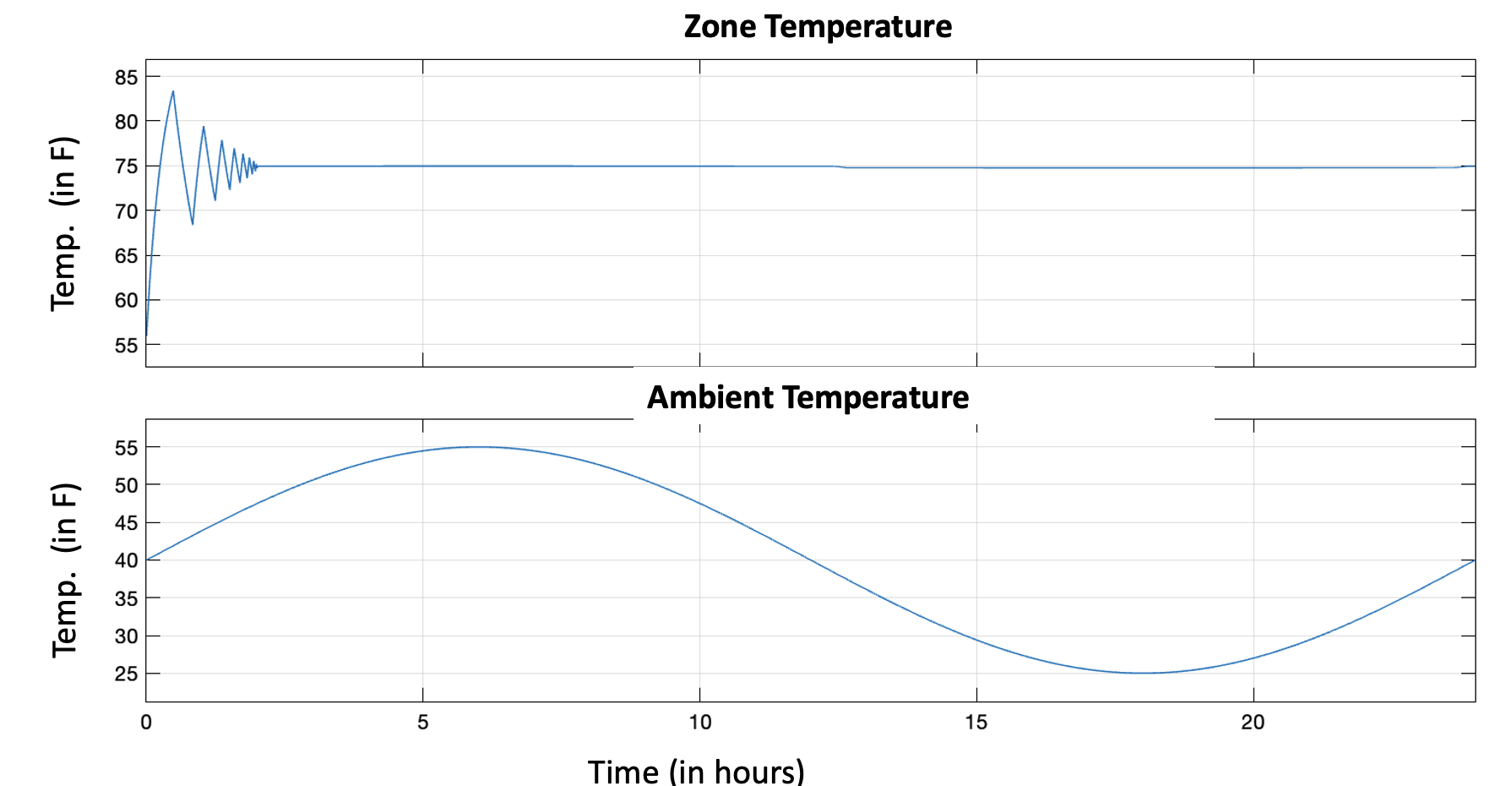}
	\caption{Temperature trajectories obtained with energy control applied to a detailed HVAC simulation model}
	\label{fig:SimulinkSim1}
\end{figure}

\subsection{Implementation-aware design}
\label{Sec:DiscreteImplementation}
The energy control in Eqn. \eqref{Eqn:ControlEnergy} together with the design implementation in Eqn. \eqref{Eqn:CtrlDynamics} may lead to frequent switching which may lead to violation of the compressor cycling constraints. 
In order to accommodate the stated objective 3 with the compressor ON-time constraint of $\tau^{min} \le \tau_{ON} \le \tau^{max}$, we incorporate a small threshold for the sliding surface leading to the design as 
\begin{equation}
        \begin{array}{l}
\dot{Q}_u ~= \alpha \qquad y_z \ge y_z^{ref} +  y_z^{+}\left(\tau^{max}\right)\\
\quad \quad  = -\alpha \qquad y_z \le y_z^{ref} + y_z^{-}\left(\tau^{min}\right)\\
\qquad  = 0\quad \qquad {\rm{otherwise}}
\end{array}
\label{Eqn:QudotThreshold}
\end{equation}
The thresholds $y_z^+$ and $y_z^-$ are computed so that the temperature excursion is between tolerable temperature bounds  $\left[T^{min},T^{max}\right]$ as permitted by the compressor cycling constraints. 
The ON time based on the thermal dynamics are derived to be related to the minimum and maximum temperature variations in \cite{lu2012evaluation} as
\begin{equation}
    \tau_{ON} = RC log\left(\frac{T^{max} - T_0 + R y_z}{T^{min} - T_0 + R y_z}\right)
\label{Eqn_Duration}
\end{equation}
Rearranging terms and assuming $ \frac{\tau_{ON}}{RC} << 1$ we obtain
\begin{equation}
    y_z = \frac{1}{R}\left(T^{min} - T0\right)\left(\frac{\tau_{ON}}{{RC}}\right) - \frac{1}{R}\left(T^{max} - T^{min}\right)
\end{equation}
With the limits of $\tau_{ON}$ given, we can obtain the permissible bounds on $y_z$ as 
\begin{subequations}
\begin{align}
    y_z^+ &= \frac{1}{R}\left(T^{min} - T0\right)\left(\frac{\tau^{max}}{{RC}}\right) - \frac{1}{R}\left(T^{max} - T^{min}\right)\\
    y_z^- &= \frac{1}{R}\left(T^{min} - T0\right)\left(\frac{\tau^{min}}{{RC}}\right) - \frac{1}{R}\left(T^{max} - T^{min}\right)
\end{align}
\label{Eqn:SigmaBound}
\end{subequations}

The responses for the same ambient temperature in Figure \ref{fig:SimulinkSim1} with the sliding mode energy control of Eqn. \eqref{Eqn:QudotThreshold} with the thresholds defined in Eqn. \eqref{Eqn:SigmaBound} with $\tau^{min} = 5$ minutes and $\tau^{max} = 15$ minutes  is shown in Figure \ref{fig:SimulinkSim2}. If the values of $R$ abd $C$ are not exactly known, we can further utilize conservative bounds in Eqn. \eqref{Eqn:SigmaBound}.
The implementation with  \eqref{Eqn:QudotThreshold} on an average corresponds to an equivalent continuous time signal of virtual control $\dot{Q}_u$ in Eqn. \eqref{Eqn:ControlEnergy} over longer timescales.

\begin{figure}[!htbp]
	\centering
	\includegraphics[width=  \linewidth]{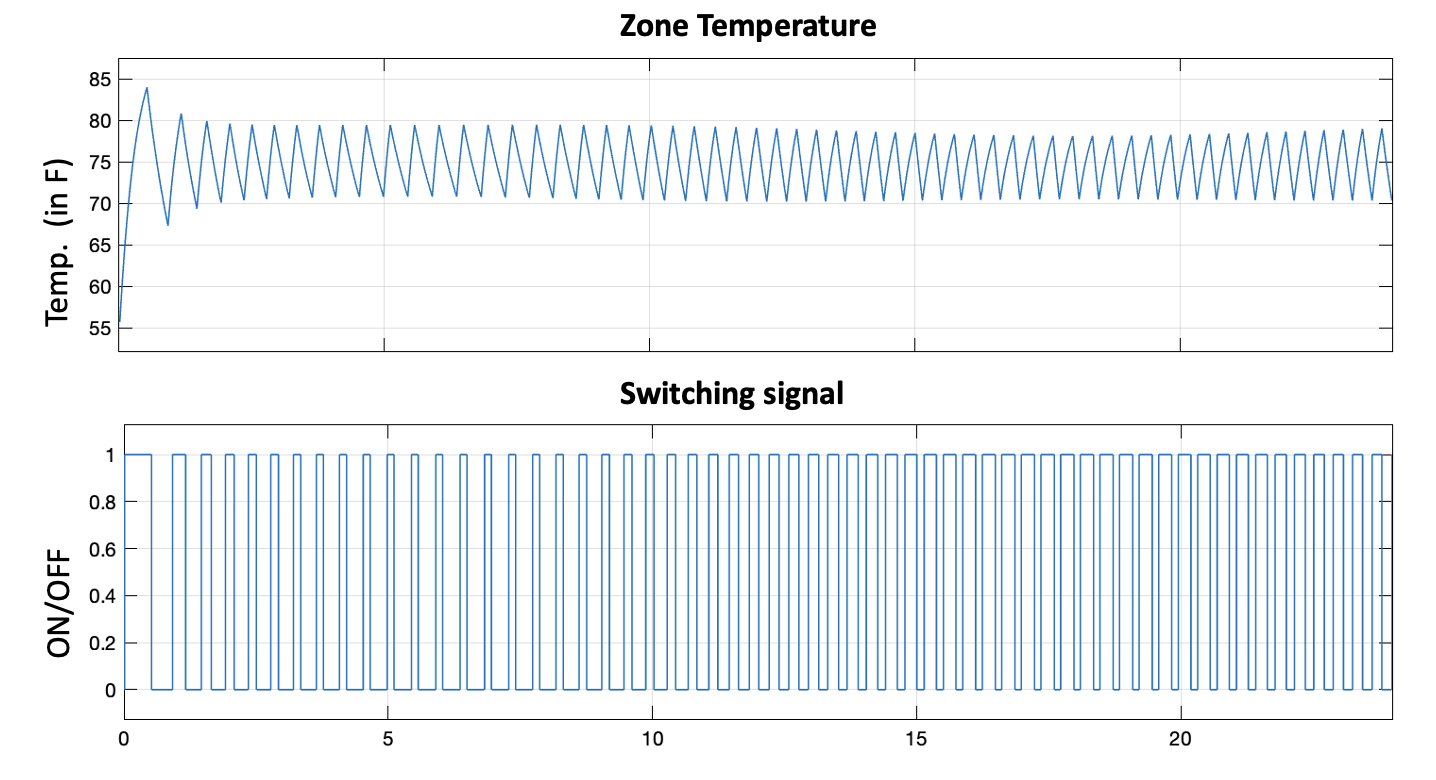}
	\caption{Temperature trajectories obtained with implementation-aware energy control applied to a detailed HVAC simulation model with compressor cycling constraints considered}
	\label{fig:SimulinkSim2}
\end{figure}

\subsection{Step-response for non-zero adjustment signal}
Next, 
we assess
if a non-zero step signal $\Delta P^{r,out}[n] = +0.2 kW$ used in Eqn. \eqref{Eqn:yzRef}
can be chased by our proposed sliding mode energy control. We activate this step signal at $t= 0.5$ hours. 

The resulting plots are shown in Fig. \ref{Fig:PrimaryCtrly} with perfect tracking performance achieved.
Notice that the response time of less than 5 seconds is achieved because of the implementation timestep chosen as $0.001$ seconds. The ramp time is clearly less than 5 minutes, which can further be improved with selection of sliding mode gain $\alpha$ utilized in Eqn. \eqref{Eqn:ControlContinuousTime}. The reserve magnitude variability tolerance of less than 5\% of the target is also achieved as is evident from the figure. Furthermore, the duration of $30$ minutes of reserves availability has also been validated.

The electrical power is also shown in Fig. \ref{Fig:PrimaryCtrlP}
Here, the base signal utilized for comparison is the one when the embedded automation does not respond to the regulation signal, i.e., the step-change in $y_z^{ref} [n]$) at 0.5 hours. After $0.5$ hours, notice that there is a difference of approximately $0.2$ kW as required by the regulation reserve signal. 

\begin{figure}[!htbp]
	\centering
	\subfigure[Output of interest $y_z$]{\includegraphics[width=0.485\linewidth]{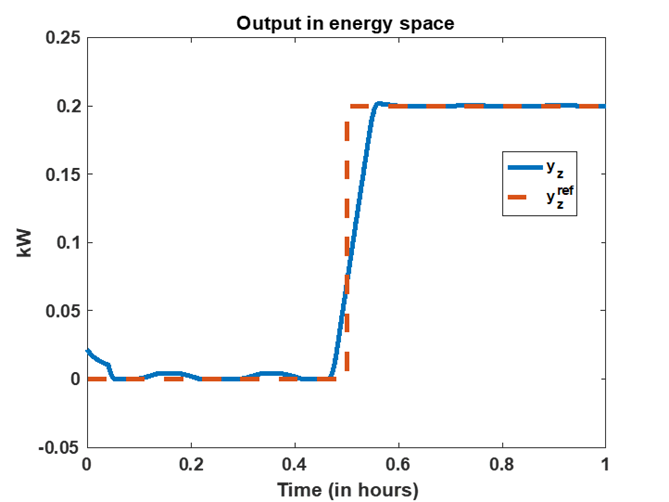}\label{Fig:PrimaryCtrly} }
	\subfigure[Averaged power consumption of HVAC with and without demand response]{\includegraphics[width=0.485\linewidth]{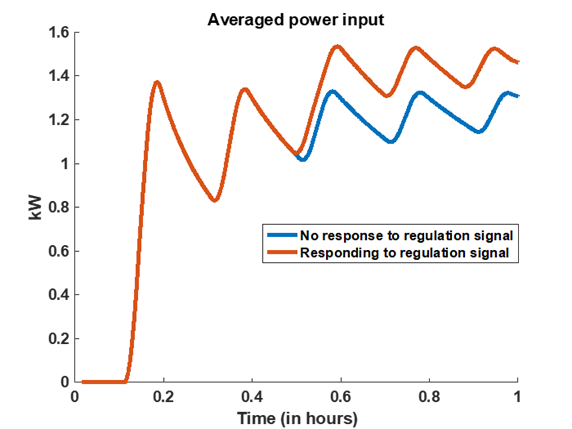}\label{Fig:PrimaryCtrlP}}
	\subfigure[Temperature $T$ of the zone $y_z$]{\includegraphics[width=0.485\linewidth]{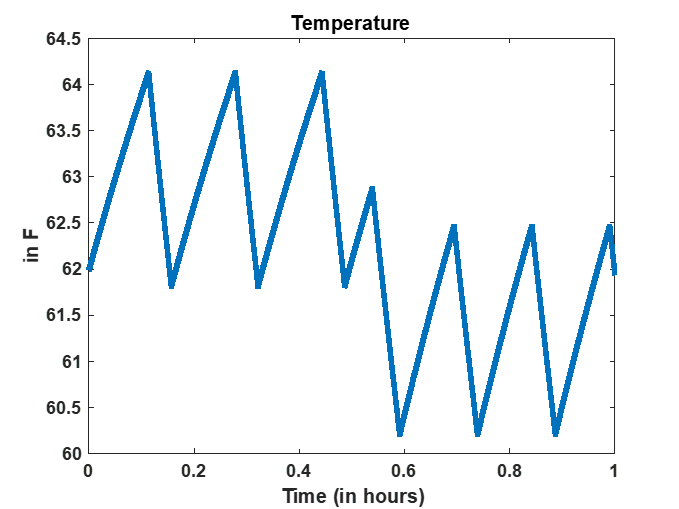}\label{Fig:PrimaryCtrlT} }
	\subfigure[Switch position $u$]{\includegraphics[width=0.485\linewidth]{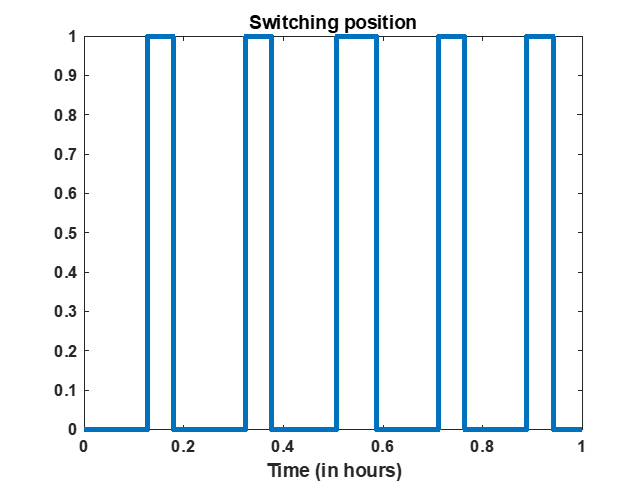}\label{Fig:PrimaryCtrlu}}
	\caption{HVAC primary control tracking output in energy space $y_z$, to the value  $y_z^{ref}$ provided by the secondary layer control}
	\label{Fig:PrimaryCtrlPower}
\end{figure}

The temperature is also ensured to be within the permissible range of $60-65 ^0 F$ as shown in Figure \ref{Fig:PrimaryCtrlT}. Finally, the imposition of the deadband in Eqn. \eqref{Eqn:QudotThreshold} of the control also ensures that the switching actions are not too frequent, as shown in Figure \ref{Fig:PrimaryCtrlu}.

We have thereby reassessed objective 5 in context of the unit testing of a single device for step change in regulation signal. Here 
we have assumed that consistent reference signal $y_z^{ref}$ is provided so that the comfort metrics are not violated. 

In order to explicitly accommodate the thermal comfort constraints, we pose a secondary control problem that dictates the adjustment signal $\Delta P^{r,out}[n]$ over faster timescale so that the step change in regulation signal of $0.2 kW$ can be better tracked over longer time horizons.  
In order to ensure that, we pose a secondary control problem formulation over relatively slower timescales.  

\section{ Efficient secondary layer control}
\label{Sec:SecondaryControl}
With the primary control explained in Section \ref{Sec:PrimaryControl}, we have ensured that the tracking of $y_z$ to $y_z^{ref}$ happens within finite time.
Once the primary dynamics settle, we can derive a quasi-static 
droop relation as stated next. 

\begin{proposition}
    Given the closed loop model in Eqn. \eqref{Eqn:ControlContinuousTime}
    there exists a 
   three-way incremental droop relation between the output in energy state space $\Delta y_z [n]$, its reference value 
   $\Delta y_z^{ref} [n]$ and the
   the power consumption $\Delta P^{r,out} [n]$  as shown in Eqn. \eqref{Eqn:Droop} where $\sigma$ is called the droop constant which is primary control dependent. 
 \begin{equation}
    \Delta y_z[n] = \left(1 -\sigma \right) \Delta y_z^{ref}[n] -\sigma\Delta P^{r,out}[n]
    \label{Eqn:Droop}
\end{equation}
\end{proposition} 
 
 These droop relations are utilized to solve an MPC problem posed over a planning time horizon of 1 minute to ensure accumulated power consumption adjustments $\Delta P^{r,out}[n]$ tack a regulation signal $P^{reg}[k]$ arriving every minute. 
 The secondary layer control problem is thereby  posed as follows: 
 \begin{subequations}
    \begin{gather}
\mathop {\min }\limits_{\Delta {y_{z}}^{ref}[n]} \left(\sum\limits_{n{T_s} = (k{T_t})}^{{n}{T_s} = (k + 1){T_t}} {{\mu ^e} {{P^{r,out}}[n]}}\right) +\qquad\qquad\qquad\notag\\
{
{\mu ^{reg}}\left|\left( 
\sum\limits_{n{T_S} = (k{T_t})}^{{n}{T_s} = (k + 1){T_t}} 
{{\Delta {P^{r,out}}[n]}}\right) - {P^{reg}}[k] \right|}  \\
{s.t.}\qquad \Delta {y_z}[n] = (1-\sigma) \Delta {y_z}^{ref}[n] - \sigma \Delta {P}^{r,out}[n]\label{Eqn:Droop2}\\
\qquad\qquad\qquad\qquad\qquad {y_z}\left(kT_t\right)= {y_{z0}}\notag\\
\begin{array}{l}
\frac{1}{R}\left(T^{ref} - T^{db} - T_0\right) \le {y_z}[n]\\
\qquad \qquad \qquad  \le \frac{1}{R}\left(T^{ref} + T^{db} - T_0\right) \label{Eqn:yzLim}\\
\qquad \forall nT_s\in \left(kT_T,(k+1)T_t\right)
\end{array}
    \end{gather}
    \label{Eqn:SecondaryCtrlPorblem}
\end{subequations}

In the droop model in Eqn. \eqref{Eqn:Droop2}, we assume the droop $\sigma$ although is operating conditions dependent, is changing much slower than $T_s$.
The droop equation forms the secondary control discrete time model with state $y_z$, control $\Delta y_z^{ref}$ and the free variable $\Delta P^{r,out}[n]$. 
$\mu^{reg}$ is the
penalty of not following the regulation signal, while $\mu^{e}$ is the fixed energy cost being paid by the device. 
As a result, the objective of secondary control is to optimize the trade-offs between energy consumption and provably supply of reserves.
In addition we incorporate thermal comfort thruogh the limits on $y_z$ in Eqn. \eqref{Eqn:yzLim} stemming from the definition of $y_z$ in Eqn. \eqref{Eqn:Output} and to incorporate the temperature bounds stated in Objective 2. 
The number of switching cycles can also be accommodated as an additional constraint based on the bounds on $y_z$ utilized for use in primary control in Eqn. \eqref{Eqn:SigmaBound}. 
The result of this optimization is the sequence of reference signals $y_z^{ref} [n]$ to which the primary control in  Eqns. \eqref{Eqn:CtrlDynamics}, \eqref{Eqn:QudotThreshold} responds to. 

\subsection{Case of single HVAC}
The interactive primary and secondary controllers are simulated for a house with HVAC with permissible temperature of with secondary control timestep of $T_s=5$ minutes and a horizon length of 1 step to obtain secondary control actions as the regulation signal arrives. The assumed costs are $\mu^{reg}=100 \$/kW$ and $\mu^e=10\$/kWh$. The trajectories obtained with this approach is compared with that when the horizon length is $12$ steps, amounting to 1 hour of planning. 

The overlaid plots of temperature are for the cases with and without MPC are shown in Fig. \ref{Fig:SecondaryCtrlT}. Notice that the temperature crosses the limits at few time instants because of the constraint on $y_z$ with upper and lower bounds serving as proxies to temperature constraints were softened by plugging them in objective function with a penalty factor of $1e2$.

Notice that the secondary control MPC problem computes the reference signals that would result in the temperatures to be within pre-specified limits. As a result, the primary control implementation would not lead to saturation. 

\begin{figure}[!htbp]
	\centering
	\subfigure[With MPC]{\includegraphics[width=0.485\linewidth]{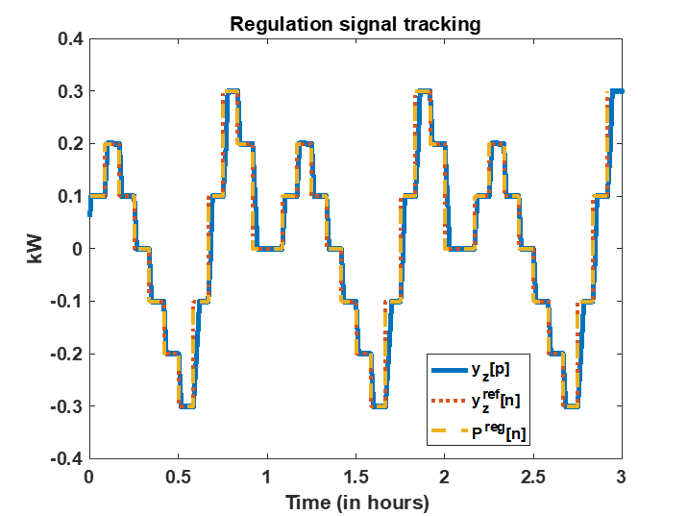}\label{Fig:SecondaryCtrlMPC}}
	\subfigure[ Without MPC]{\includegraphics[width=0.485\linewidth]{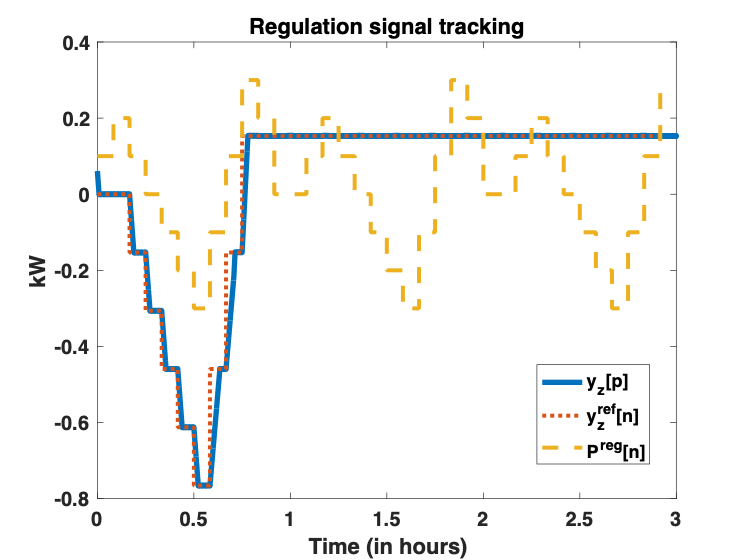}\label{Fig:SecondaryCotrnolNoMPC} }
	\caption{Tracking of regulation signal with interactive primary and secondary control of HVAC system}
\end{figure}

\begin{figure}[!htbp]
	\centering
	{\includegraphics[width=0.7\linewidth]{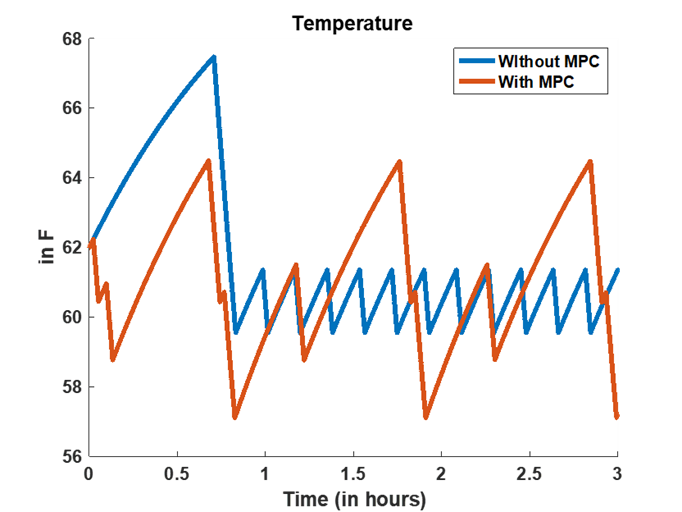} }
	\caption{Temperature trajectories for the cases with and without MPC-based secondary control interacting with primary control of HVAC system $y_z$. 
	}
	\label{Fig:SecondaryCtrlT}
\end{figure}

Fig. \ref{Fig:SecondaryCtrlMPC} and \ref{Fig:SecondaryCotrnolNoMPC} respectively show the result of the control action for the case with and without MPC respectively . In these plots, the primary control action results in fast-changing values of $y_z$ shown in blue perfectly chasing the reference signal $y_z^{ref} [n]$ in red computed every $5$ minutes.  Overlaid is also the regulation signal, which is approximately same as the secondary control action $y_z^{ref}[n]$ since $\Delta y_z[n]$ is ensured to remain close to zero to minimize temperature deviations. 
Notice from Figure \ref{Fig:SecondaryCtrlMPC} that the case with MPC results in perfect tracking of the regulation signal while the case without MPC emphasizes more on minimizing objectives at current time instant without any foresight of future temperature trajectories. 

\subsection{Case of 50 HVACs}
In the previous subsection, we softened the temperature constraint to prevent infeasibility of the posed secondary control problem in Eqn. \eqref{Eqn:SecondaryCtrlPorblem}. However, this results in unacceptable temperature deviations. To make this a hard constraint, it becomes important to consider more HVAC units. 

We thus consider next the problem of tracking the regulation signal by a group of 50 HVAC units with the interactive primary and secondary control as shown in Figure \ref{fig:HVAC_Prop}. A simple way to accommodate multiple HVAC units for secondary control is to duplicate the constraints in Eqn. \eqref{Eqn:Droop} and \eqref{Eqn:yzLim} for each HVAC unit and replace the objective for single HVAC power consumption $P^{r,out}[n]$ with the power consumption of all HVAC units i.e. summation of $P^{r,out}[n]$ of all HVAC units. 


We utilize real data taken from Mueller community of Pecan Street Inc (PSI) to obtain a regulation signal representative of the disturbances created by renewable as seen by the Austin node of the transmission grid modeled by ERCOT, the independent transmission system operator in Texas. 
We further utilize the parameters of 50 HVAC units for simulations and the temperature preferences of the respective units from the survey data as was made available by PSI to carry out this experiment. 

Shown in Fig. \ref{Fig:HVACIndi} are the individual power consumption trajectories of randomly selected HVAC units. 
50 HVACs together adjust their power consumption from their baseline values as shown in red in Fig.\ref{Fig:HVACAggreg} overlaid on top of the commanded regulation signal in blue. Notice that the perfect tracking was achievable even with a small number of HVAC units. 
\begin{figure}[!htbp]
	\centering
	{\includegraphics[width=0.9\linewidth]{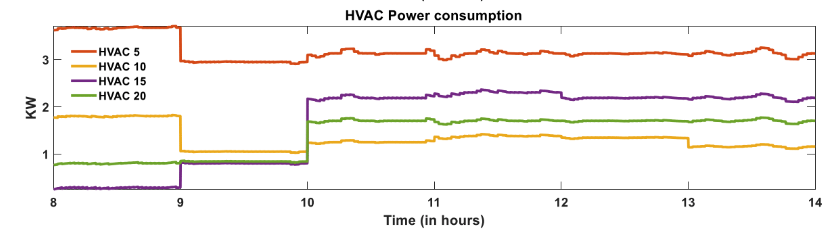} }
	\caption{Power consumption of randomly selected HVAC units
	}
	\label{Fig:HVACIndi}
\end{figure}

\begin{figure}[!htbp]
	\centering
	{\includegraphics[width=0.9\linewidth]{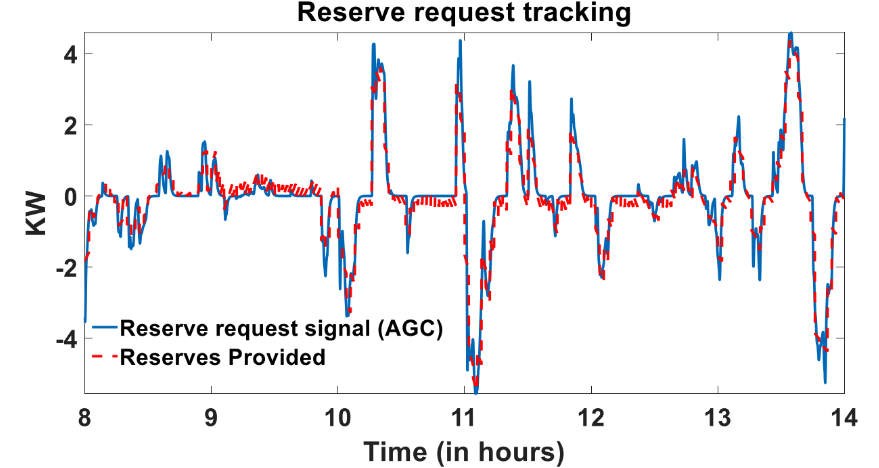} }
	\caption{Aggregate performance of 50 HVACs tracking the reserve signal
	}
	\label{Fig:HVACAggreg}
\end{figure}



\section{Conclusion}
\label{Sec:Conclusions}
The contributions of this paper are multi-fold.  First,  a novel energy-based  approach  is introduced to facilitate model-free control design utilizing minimal sensor measurement data for implementation.  Using this modeling,  a  robust sliding mode control  of energy dynamics is derived, along with theoretical conditions under which proposed control design results in a provable three-way droop relation between output variable of interest, secondary control signals and the desired electrical power adjustments over slower time scales.  It is then shown how a  secondary layer model predictive controller accommodates provable droop relations, performance metrics and implementation constraints. Finally, simulation-based evidence is presented  for backing up claims on efficiency and provability of the proposed design method using real data and practical consideration of a HVAC unit in one community in Texas. In particular, it is shown that the controlled real power consumed by the HVAC comprises two components: one that does real work and the other corresponding to wasted work. The former compensates for the effect of energy imbalances created while participating in the provision of regulation reserves.  The latter corresponds to the wasted energy required to compensate for the effect of thermal consumption created due to the temperature set-point adjustment for tracking the regulation reserve signal.  
Notably, this first-of-its-kind relation between electrical and thermal processes is viewed as being fundamental to having sufficient granularity of the model. Further work is needed toward deploying these controllers and fully utilizing their potential.


%

\appendices

\section{Proof of Proposition 1}
\label{Sec:AppProof}
\begin{proof}
We first derive the expressions for energy variables in context of the simplified dynamical models in use in Eqn. \eqref{Eqn_ConvModel}. 
In Definition \ref{Defn:Energy}, the entropy is a complex function of temperature given by $S = K(T)$, which leads to the simplification for stored energy as in Eqn. \eqref{Eqn_IntEnergyAlt}
Since entropy and entropy flow are abstract concepts and are not physically measurable, empirical relations are often utilized. The term  ${\left({T-T_0}\right)\frac{{\partial K(T)}}{{dT}}}$ is assumed constant and is called the thermal capacitance  $C_w$ used in the simplified model in Eqn. \eqref{Eqn_ConvModel}, thereby resulting in following commonly used relation 
\begin{subequations}
	 \begin{equation}
	 E \approx U = \int_0^t {T\frac{{\partial K(T)}}{{dT}}\frac{{dT}}{{dt}}dt}  = \underbrace {\frac{{\partial K(T)}}{{dT}}{\left( {T - {T_0}} \right)}}_{{C_w}}\left( {T - {T_0}} \right)
	 \label{Eqn_IntEnergyAlt}
	 \end{equation}

Here, we have assumed initial temperature is the ambient temperature $T_0$. Next, the stored energy in tangent space in Eqn.  \eqref{Eqn:EtEnergy} can similarly be simplified as 
\begin{equation}
    \begin{array}{l}
{E_t} = \int_0^t {\frac{{dT}}{{dt}}\frac{d}{{dt}}\left( {\frac{{\partial K(T)}}{{dT}}\frac{{dT}}{{dt}}} \right)dt}  = \\
 = \int_0^t {\frac{{\partial K(T)}}{{dT}}\frac{{dT}}{{dt}}\frac{{{d^2}T}}{{d{t^2}}}dt}  + \int_0^t {\frac{d}{{dt}}\left( {\frac{{\partial K(T)}}{{dT}}} \right){{\left( {\frac{{dT}}{{dt}}} \right)}^2}dt} \\
 = \frac{1}{2}\frac{{\partial K(T)}}{{dT}}{\left( {\frac{{dT}}{{dt}}} \right)^2} = \frac{{{C_w}}}{{2\left(T-T_0\right)}}\left( {\frac{{dT}}{{dt}}} \right)^2
\end{array}
    \label{Eqn:EtExp}
\end{equation}
In the second equation above, we have assumed the term ${\frac{{\partial K(T)}}{{dT}}}$ is time independent. We have further assumed the time derivative of temperature at initial time is zero.
For the simplified model, the injected power as viewed from the perspective of thermal energy domain is given as in Eqn. \eqref{Eqn_Psupply}. However it also must satisfy the definition of instantaneous power in Eqn. \eqref{Eqn_RealReacP}, thereby resulting in: 
\begin{equation}
    \begin{array}{l}
{S_f}\left( {T - {T_0}} \right) = {{\dot m}_a}{C_p}\left( {{T_{sup}} - T} \right)\\
 \Rightarrow {S_f} = \left( {{{\dot m}_a}\frac{{{C_p}}}{{T - {T_0}}}} \right)\left( {{T_{sup}} - T} \right)
\end{array}
\end{equation}

Here, we assume the term $S_f$ is time-independent. 
With this effort flow demarcation for the thermal processes involved, we can derive the rate of reactive power based on Eqn. \eqref{Eqn_RealReacP} as:
\begin{equation}
    \begin{array}{l}
{{\dot Q}_T} = \left( {T - {T_0}} \right)\frac{d}{{dt}}\left( {\frac{{{{\dot m}_a}{C_p}}}{{T - {T_0}}}\left( {{T_{\sup }} - T} \right)} \right)\\
\qquad  - \left( {\frac{{{{\dot m}_a}{C_p}}}{{T - {T_0}}}\left( {{T_{\sup }} - T} \right)} \right)\frac{d}{{dt}}\left( {T - {T_0}} \right)
\end{array}
\label{Eqn:QTExp}
\end{equation}
\end{subequations}

\begin{subequations}
Defining the damping losses $\frac{E}{\tau}$ with the time constant $\tau = \frac{R}{C_w}$, we can derive the 
the first equation in Eqn. \eqref{Eqn_GeneralEnergyModel} starting from the thermal model in Eqn. \eqref{Eqn_ConvModel} as follows
\begin{equation}
    \dot{E} = C_w \frac{dT}{dt} = -\frac{1}{R}\left(T - T_0\right) + P_u = -\frac{E}{\tau} + P_u
\end{equation}

Now let us denote the integrand in Eqn. \eqref{Eqn_IntEnergyAlt} as rate of change of stored energy $p$. Taking its time derivative, we obtain the following simplification
\begin{equation}
    \begin{array}{l}
\dot p = \frac{{\partial K(T)}}{{dT}}{\left( {\frac{{dT}}{{dt}}} \right)^2} + \frac{{\partial K(T)}}{{\partial T}}\left( {T - {T_0}} \right)\frac{{{d^2}T}}{{d{t^2}}}\\
 = 2{E_t} + {C_w}\frac{{{d^2}T}}{{d{t^2}}}
\end{array}
\end{equation}
Here we used the expression for stored energy in tangent space from Eqn. \eqref{Eqn:EtExp}.
Expanding the second term using the model in Eqn. \eqref{Eqn_ConvModel}, we obtain
\begin{equation}
    \begin{array}{l}
\dot p = 2{E_t} + \frac{d}{{dt}}\left( {{{\dot m}_a}{C_p}\left( {{T_{\sup }} - T} \right) - \frac{1}{R}\left( {T - {T_0}} \right)} \right)\\
 = 2{E_t} + \frac{d}{{dt}}\left( {{S_f}\left( {T - {T_0}} \right)} \right) - \frac{1}{R}\frac{d}{{dt}}\left( {T - {T_0}} \right)\\
 = 2{E_t} + \left( {\left( {T - {T_0}} \right)\frac{{d{S_f}}}{{dt}} - {S_f}\frac{d}{{dt}}\left( {T - {T_0}} \right)} \right)\\
 + {S_f}\frac{{dT}}{{dt}} - \frac{1}{R}\frac{{dT}}{{dt}}\\
 = 2{E_t} + {{\dot Q}_T} + \frac{{{C_w}}}{{T - {T_0}}}{\left( {\frac{{dT}}{{dt}}} \right)^2}
 = 4{E_t} + {{\dot Q}_T}
\end{array}
\label{Eqn:Temp22}
\end{equation}

Also for a switching device, considering the left hand side and right-hand side reactive power flows $\dot{Q}_{sw,L},\dot{Q}_{sw,R}$ as depicted in Figure \ref{Fig:QdotSwitch}, we have
\begin{equation}
\dot{Q}_{sw,R} = \left(
    \dot{Q}_{sw,L} + \dot{Q}_{u}\right)
    \label{Eqn:QBal}
\end{equation}
Also, since $\dot{Q}_{sw,L} = \dot{Q}^{r,out} $ and $\dot{Q}_{sw,R} = -\dot{Q}_T$ in this case, we can further simplify Eqn. \eqref{Eqn:Temp22} as 
\begin{eqnarray}
    \dot{p} &= 4E_t - \dot{Q}_{sw,R} 
    = 4E_t - \left(\dot{Q}_{sw,L} + \dot{Q}_u\right)\notag\\
    &=4E_t - \dot{Q}^{r,out} - \dot{Q}_u 
\end{eqnarray}
\end{subequations}
\end{proof}

\section{Proof of Theorem  1}
\label{Sec:ProofTheorem}
\begin{proof}
\textit{(1)}:
From the expression for stored energy in Eqn. \eqref{Eqn_IntEnergyAlt}, $E = C_w \left(T - T_0\right)$ and $p = \dot{E} = -\frac{1}{R}\left(T - T_0\right) + P^{rated}u =
-\frac{1}{R}\left(\frac{E}{C_w} - T_0\right) + + P^{rated}u
$. Clearly, there is a one-one mapping from $\left[T, u\right]$ to $\left[E,p\right]$. We can thereby utilize energy space model in place of the conventional state space model for the rest of analysis and control design.
Even if the AHU dynamics were modeled, we could establish diffeomporshism between states appearing at the interfaces and the energy space variables. For details, refer to \cite{jaddivada2020unified}. 

\textit{(2)\&(3)}:
Let $\sigma = y_z - y_z^{ref}$,
Starting from Eqn. \eqref{Eqn:Output} and \eqref{Eqn:yzRef} and taking the time derivative we have 
\begin{equation}
    \begin{array}{l}
\frac{d\sigma}{dt} = \frac{d}{{dt}}\left( {{y_z} - {y_z}^{ref}} \right) =  - \frac{1}{R}\frac{{dT}}{{dt}} - {{\dot P}_u}\\
 =  - \frac{d}{{dt}}\left( {\frac{1}{R}\left( {T - {T_0}} \right) + {P_u}} \right)\\
 =  - {C_w}\dot T =  - \dot p =  - \left( {4{E_t} + 2{{\dot Q}_T}} \right) + {{\dot Q}_u}
\end{array}
\end{equation}
Plugging in the control design in Eqn. \eqref{Eqn:ControlEnergy}, we obtain 
\begin{equation}
    {\sigma} = \le \left(\overline{L} - \alpha\right) sign(\sigma) \le -K sign(\sigma)
\end{equation}

Consider next a storage function $V = \frac{1}{2} \sigma^2$.
We obtain 
\begin{equation}
    \dot{V} = \sigma \dot{\sigma} \le -K\left|\sigma\right|  = -K(2V)^{1/2}
    \label{Eqn:Temp11}
\end{equation}
From this expression, clearly $\dot{V} \le 0$ for positive gain $K$ and by LaSalle's invariance theorem, we obtain the desired result. 
Furthermore, taking the time integral of Eqn. \eqref{Eqn:Temp11}, we obtain 
\begin{equation}
    \begin{array}{l}
{V^{1/2}}(t) - {V^{1/2}}(0) \le  - \sqrt 2 Kt\\
 \Rightarrow |\sigma (t)| - |\sigma (0)| \le  - \sqrt 2 Kt
\end{array}
\end{equation}
For $\sigma(t) \rightarrow 0$, rearranging terms, we see that the maximum reaching time is $t_r \le \frac{\sigma (0)|}{\sqrt{2} K}$
\end{proof}

\section{Proof of Proposition 2}
\label{Sec:DroopProof}
\begin{subequations}
Consider the closed loop dynamical model in Eqn. \eqref{Eqn_ConvModel}, \eqref{Eqn:ControlEnergy} and \eqref{Eqn:CtrlDynamics}. Let us denote $P' = P^{rated} u$ in Eqn. \eqref{Eqn_ConvModel} and let us characterize its dynamics by reexpressing rate of reactive power in terms of rate of instantaneous power following the definitions in Eqn. \eqref{Eqn_RealReacP}. We thereby obtain
\begin{equation}
    \begin{array}{l}
\dot T =  - \frac{1}{{R{C_w}}}(T - {T_0}) + \frac{1}{{{C_w}}}{P'}\\
{{\dot P}'} = 2\frac{{{P'}}}{{T - {T_0}}}\frac{{dT}}{{dt}} - \dot{Q}_{sw,L} \\
=2\frac{{{P'}}}{{T - {T_0}}}\frac{{dT}}{{dt}} + \alpha sign(\sigma ) -\dot{Q}_{T} 
\end{array}
\end{equation}
The second equation simplification is a result of Eqn. \eqref{Eqn:QBal} and further substuting of the expression of $\dot{Q}_u$ from Eqn. \eqref{Eqn:ControlEnergy}. 
Simplifying the resulting equation further by substituting the temperature dynamics, we obtain
\begin{equation}
    {{\dot P}'} = \underbrace {\left( {\frac{{ - 2}}{{R{C_w}}}} \right)}_a{P'} + \underbrace {\frac{\alpha }{{|\sigma |}}}_{b(\tilde x)}\sigma  + \underbrace {2\frac{1}{{{C_w}(T - {T_0})}}{P'}^2 + \dot{Q}_T}_{c(\tilde x)}
    \label{Eqn:Temp0}
\end{equation}
At time t = $nT_s$ taking the time derivative to zero. Repeating the same at $t = (n-1)T_s$ and then taking the difference between the two equations, we obtain
\begin{equation}
    a\Delta {P'}[n] = b(\tilde x)\Delta \sigma [n]
    \label{Eqn:Temp1}
\end{equation}
Here, $\Delta \left(.\right)[n] = \left(.\right)(nT_s) - \left(.\right)((n-1)T_s)$.
We have ignored the incremental effects of $c(\tilde{x})$ in Eqn. \eqref{Eqn:Temp0} to obtain Eqn. \eqref{Eqn:Temp1}. We have also assumed the coefficient $b(\tilde{x})$ is almost constant between subsequent $nT_s$ timesteps. 
Since $\sigma = y_z - y_z^{ref}$,
we can establish the following relation
\begin{equation}
    \Delta \sigma[n] =  \Delta y_z[n] - \Delta y_z^{ref} [n]
    \label{Eqn:Temp2}
\end{equation}
Furthermore, from Eqn. \eqref{Eqn:yzRef}, we can write an incremental relation as 
\begin{equation}
    \Delta P' [n] = \Delta y_z^{ref}[n]
    - \Delta P^{r,out}[n]
    \label{Eqn:Temp3}
\end{equation}
Substituting Eqn. \eqref{Eqn:Temp3} and \eqref{Eqn:Temp2} into Eqn. \eqref{Eqn:Temp1}, we obtain
\begin{equation}
    a(\Delta {y_z}^{ref}[n] - \Delta {P^{reg}}[n]) = b(\tilde x)\left( {\Delta {y_z}[n] - \Delta {y_z}^{ref}[n]} \right)
\end{equation}
Rearranging the terms, we obtain
\begin{equation}
    \Delta {y_z}[n] = \underbrace {\left( {1 - b{{(\tilde x)}^{ - 1}}a} \right)}_{1 - \sigma (\tilde x)}\Delta {y_z}^{ref}[n] - \underbrace {\left( {b{{(\tilde x)}^{ - 1}}a} \right)}_{\sigma (\tilde x)}\Delta {P^{reg}}[n]
\label{Eqn:DroopExpansion}
\end{equation}
\end{subequations}

\begin{remark}
    Notice the similarity of HVAC droop relation in Eqn. \eqref{Eqn:DroopExpansion} to the one typically utilized for generation resources \cite{ilic2011possible,ilic2000dynamics}
    \begin{equation}
        \Delta \omega[n] = \left(1 -\sigma_G D_G\right) \Delta \omega^{ref}[n] -\sigma_G \Delta P^{reg}[n]
    \end{equation}
    Here $\sigma_G$ is the generator droop and $D$ is the rotor damping. 
\end{remark}
\section*{Acknowledgment}
 DISTRIBUTION STATEMENT A. Approved for public release. Distribution is unlimited.
 \\
 
This material is based upon work supported by the Department of Energy
under Air Force Contract No. FA8702-15-D-0001. Any opinions, findings,
conclusions or recommendations expressed in this material are those of the
author(s) and do not necessarily reflect the views of the Department of
Energy.
\\

\copyright 2021 Massachusetts Institute of Technology
\\

Delivered to the U.S. Government with Unlimited Rights, as defined in
DFARS Part 252.227-7013 or 7014 (Feb 2014). Notwithstanding any copyright
notice, U.S. Government rights in this work are defined by DFARS
252.227-7013 or DFARS 252.227-7014 as detailed above. Use of this work
other than as specifically authorized by the U.S. Government may violate any
copyrights that exist in this work.

\ifCLASSOPTIONcaptionsoff
  \newpage
\fi



\bibliographystyle{unsrt}
\bibliography{IEEEabrv,HVAC}

%









\end{document}